\pdfoutput=1
\documentclass[11pt]{article}
\newif\ifdraft
\newif\ifhideproofs
\drafttrue

\usepackage[letterpaper,margin=1in]{geometry}
\usepackage{bm, amsmath, amssymb, amsthm, amsfonts}
\usepackage{hyperref}
\usepackage[capitalize,nameinlink]{cleveref} 
\usepackage{cite}
\usepackage{framed}
\usepackage[framemethod=tikz]{mdframed}
\usepackage{appendix}
\usepackage{graphicx}
\usepackage{color}
\usepackage{varwidth}
\usepackage{wrapfig}
\usepackage{algorithm2e}
\usepackage[noend]{algpseudocode}
\usepackage[textsize=tiny]{todonotes}

\usepackage{enumitem}
\setitemize{noitemsep,topsep=3pt,parsep=3pt,partopsep=3pt}

\usepackage{xspace}
\usepackage{xifthen}

\definecolor{darkgreen}{rgb}{0,0.5,0}
\definecolor{darkblue}{rgb}{0,0,0.8}
\usepackage{hyperref}
\usepackage{thmtools}
\usepackage{thm-restate}

\makeatletter
\g@addto@macro\bfseries{\boldmath}
\makeatother

\hypersetup{
	colorlinks=true,                  
	linkcolor=darkblue,
	citecolor=darkgreen
}

\declaretheorem[name=Theorem,numberwithin=section]{theorem}
\newtheorem{lemma}[theorem]{Lemma}
\newtheorem{corollary}[theorem]{Corollary}

\newtheorem{claim}[theorem]{Claim}

\newtheorem{definition}{Definition}[section]

\newcommand{\mypara}[1]{
\smallskip

\noindent\textbf{#1}}

\newcommand{\NF}{N^F}
\newcommand{\NNF}{N^{\overline{F}}}

\newcommand{\Deltab}{\tilde{\Delta}}

\DeclareMathOperator*{\argmin}{arg\,min}

\newcommand{\ignore}[1]{}

\algnewcommand\algorithmicswitch{\textbf{switch}}
\algnewcommand\algorithmiccase{\textbf{case}}

\algdef{SE}[SWITCH]{Switch}{EndSwitch}[1]{\algorithmicswitch\ #1\ \algorithmicdo}{\algorithmicend\ \algorithmicswitch}%
\algdef{SE}[CASE]{Case}{EndCase}[1]{\algorithmiccase\ #1}{\algorithmicend\ \algorithmiccase}%
\algtext*{EndSwitch}%
\algtext*{EndCase}%

\newcommand{\CONGEST}{\ensuremath{\mathsf{CONGEST}}\xspace}
\newcommand{\LOCAL}{\ensuremath{\mathsf{LOCAL}}\xspace}

\newcommand{\PSLOCAL}{\ensuremath{\mathsf{P-SLOCAL}}\xspace}

\newcommand{\eps}{\varepsilon}
\newcommand{\poly}{\operatorname{\text{{\rm poly}}}}

\newcommand{\set}[1]{\left\{#1\right\}}
\newcommand{\logstar}[1]{\log^{*} #1}
\DeclareMathOperator{\E}{\mathbb{E}}
\DeclareMathOperator{\polylog}{\poly\log}

\newcommand{\CC}{\mathcal{C}}

\newcommand{\complexityclass}[2][]{\ensuremath{\mathsf{#2}\ifthenelse{\isempty{#1}}{}{(#1)}}}

\setlist[enumerate]{itemsep=0mm}

\newcommand{\hide}[1]{}

\newcommand{\FullOrShort}{short}

\ifthenelse{\equal{\FullOrShort}{full}}{
	
  \newcommand{\fullOnly}[1]{#1}
  \newcommand{\shortOnly}[1]{}

}{

  \newcommand{\shortOnly}[1]{#1}
  \newcommand{\fullOnly}[1]{}

}

\ifhideproofs
  \usepackage{environ}
  \NewEnviron{beweisweg}{}

\fi

\begin{document}

\newcommand{\auth}[3]{\textbf{#1}\par#2\par#3\par\medskip}

\date{}

\title{Deterministic Distributed Dominating Set Approximation\\
in the CONGEST Model}


\maketitle

\vspace*{-1.5cm}

\auth{Janosch Deurer}
{University of Freiburg, Germany}
{deurer@cs.uni-freiburg.de}

\auth{Fabian Kuhn\footnote{Partly supported by ERC Grant No.\ 336495 (ACDC).}}
{University of Freiburg, Germany}
{kuhn@cs.uni-freiburg.de}
\auth{Yannic Maus\footnotemark[1]}
{Technion, Haifa, Israel}
{yannic.maus@cs.technion.ac.il}

\begin{abstract}
  \normalsize
  We develop deterministic approximation algorithms for the minimum
  dominating set problem in the \CONGEST model with an almost optimal
  approximation guarantee. For $\eps>1/\polylog \Delta$ we obtain 
  two  algorithms with approximation factor $(1+\eps)(1+\ln (\Delta+1))$ and
  with runtimes $2^{O(\sqrt{\log n \log\log n})}$ and
  $O(\Delta\polylog \Delta +\polylog \Delta \logstar n)$, respectively.
  Further we show how dominating set approximations can be
  deterministically transformed into a connected dominating set in the
  \CONGEST model while only increasing the approximation guarantee by
  a constant factor.  This results in a deterministic
  $O(\log \Delta)$-approximation algorithm for the minimum connected
  dominating set with time complexity
  $2^{O(\sqrt{\log n \log\log n})}$.
\end{abstract}

\newpage
\section{Introduction \& Related Work}
Given a graph $G=(V,E)$, a dominating set $S\subseteq V$ of $G$ is a set of nodes of $G$ such that any node $u\not\in S$ has at least one neighbor $v\in S$. The minimum dominating set (MDS) problem---the objective is to compute a dominating set with minimal cardinality---and the equivalent minimum set cover problem are among the oldest classic combinatorial optimization problems \cite{karp72,johnson74,garey79}. Motivated by potential practical applications (e.g., as a tool for clustering in wireless ad hoc and sensor networks), the minimum dominating set problem has also found  attention in a distributed setting. In the last two decades, distributed approximation algorithms for minimum dominating set and closely related problems such as minimum set cover, minimum connected dominating set, or minimum vertex cover have been studied intensively (see, e.g., \cite{jia02,dubhashi05,constanttime,Kuhn06,lenzen08,LenzenPW13,ghaffari14_CDS,kuhn16_jacm,SLOCAL17}).

It is well known that for graphs with maximum degree $\Delta$, in the standard, sequential setting, the basic greedy algorithm achieves a $\ln(\Delta+1)$-approximation \cite{johnson74} and that up to lower order terms, this is the best that can be achieved in polynomial time unless $\mathsf{P}=\mathsf{NP}$ \cite{dinursteurer14}. The first efficient distributed approximation algorithm was given by Jia, Rajaraman, and Suel~\cite{jia02} and it achieves an approximation ratio of $O(\log\Delta)$ in time $O(\log^2 n)$. The algorithm is randomized and it works in the \CONGEST model (i.e., with messages of size $O(\log n)$).\footnote{For a formal definition of the distributed communication models we use, we refer to \Cref{sec:model}.} These bounds were improved by Kuhn, Moscibroda, and Wattenhofer \cite{Kuhn06}, who showed that it is possible to compute an $(1+\eps)(1+\ln(\Delta+1))$-approximation in time $O(\log^2\Delta / \eps^4)$ in the \CONGEST model and in time $O(\log n / \eps^2)$ in the \LOCAL model. Recently, in \cite{SLOCAL17}, Ghaffari, Kuhn, and Maus showed that, in the \LOCAL model, at the cost of allowing exponential local computations at the nodes, it is even possible to compute a $(1+o(1))$-approximation in a polylogarithmic number of rounds.

All the above distributed algorithms critically depend on the use of \emph{randomization} for breaking symmetries in an efficient manner. The only existing efficient \emph{deterministic} distributed algorithms for approximating MDS are based on a generic tool known as network decomposition \cite{awerbuch89,linial93,panconesi95,DISC18_DomSet}. These algorithms only work in the \LOCAL model (i.e., they require large messages) or they only achieve a suboptimal approximation ratio. Concretely, in the \LOCAL model, it is possible to compute a $(1+o(1))$-approximation\footnote{If internal computations at the nodes are restricted to polynomial time, it is possible to obtain a $(1+o(1))\ln(\Delta+1)$-approximation in the same asymptotic time.} of MDS in time $2^{O(\sqrt{\log n})}$ \cite{SLOCAL17} and in the \CONGEST model, an $O(\log^2 n)$-approximation can be computed in time $2^{O(\sqrt{\log n\log\log n})}$ \cite{DISC18_DomSet}. Note that these deterministic algorithms are exponentially slower than their randomized counterparts. This general behavior is true for many important distributed graph problems such as, e.g., the problems of computing a vertex coloring with $\Delta+1$ or even $O(\Delta)$ colors, the problem of computing a maximal independent set (MIS), or the problem of computing various kinds of (in some sense sparse) decompositions of the network graph into clusters of small diameter \cite{barenboimelkin_book,SLOCAL17}. While there are extremely efficient (at most polylogarithmic time) and often very simple randomized algorithms, the best deterministic algorithms are exponentially slower and require $2^{O(\sqrt{\log n})}$. Typically, these deterministic algorithms also abuse the power of the \LOCAL model by using network decompositions in a pretty brute-force kind of way: For graphs of small diameter, the algorithms essentially reduce to collecting the whole topology at a single node---this cannot be done in the \CONGEST model---and compute the solution in a centralized way, often through solving NP-complete problems.  Whether an exponential gap between the best randomized and deterministic algorithms for these problems is really necessary is considered to be one of the key open questions in the area of distributed graph algorithms \cite{barenboimelkin_book,SLOCAL17,linial92,panconesi95,MausThesis}. In the context of this more general question, Ghaffari, Harris, and Kuhn recently showed in \cite{newHypergraphMatching} that also the MDS problem has a key role in this more general question. By using (and extending) a framework developed in \cite{SLOCAL17}, they showed the MDS problem is \PSLOCAL complete, meaning: If there is a polylogarithmic-time distributed deterministic algorithm that computes any polylogarithmic approximation for the MDS problem, there are polylogarithmic-time deterministic distributed algorithms for essentially all\footnote{This in particular includes all problems where a solution can be locally checked in polylogarithmic time and all problems for which efficient distributed Las Vegas algorithms exist. It also includes many optimization problems such approximating MDS or maximum independent set.} problems for which there are efficient randomized algorithms.

\subsection{Contributions}

In the present paper, we give the first deterministic distributed algorithms for the
minimum dominating set problem in the \CONGEST model, which are efficient in terms of their time complexity and at the same time obtain an essentially optimal approximation ratio. Our first result is 
based on first computing a network decomposition by an algorithm of
\cite{awerbuch89,DISC18_DomSet}. Unlike in the \LOCAL model, where
partial solutions of clusters of small diameter can be computed in
time linear in the cluster diameter by simply collecting the whole
cluster topology at one node, in the \CONGEST model, we have to be
much more careful. Inspired by recent work on other local problems
\cite{DISC17_MIS,DISC18_DomSet}, we design a randomized \CONGEST
algorithm that  only requires $\polylog n$-wise independent randomness. Then we use the method of conditional expectations to efficiently
derandomize this algorithm. We note that the method of conditional expectations has been quite intensively used in the last couple of years in order to derandomize randomized distributed graph algorithms (see, e.g., \cite{DISC17_MIS,newHypergraphMatching,harris18,kawarabayashi18,parter18_coloring,parter18_spanners}).
\begin{theorem} \label{thm:mainThmN} For any $\eps>1/\polylog \Delta$,
  there exists a deterministic \CONGEST algorithm that computes an
  $(1+\eps)\cdot (1+\ln (\Delta+1))$-approximation for MDS in $2^{O(\sqrt{\log n\log\log n})}$ rounds. 
\end{theorem}\

The above theorem improves a result by Ghaffari and Kuhn in \cite{DISC18_DomSet}, where it was shown that, in the same asymptotic time complexity, it is possible to compute an $O(\log^2 n)$-approximation for the MDS problem. We note that while our algorithm is based on carefully and gradually rounding a given fractional solution to an integer solution (cf.\ \Cref{sec:nutshell}), the algorithm of \cite{DISC18_DomSet} is based on a relatively direct parallel implementation of the sequential greedy algorithm, where the decision of which nodes to add in each greedy step is achieved by derandomizing a basic distributed hitting set problem.

The time complexity of our algorithm is $2^{O(\sqrt{\log n \log \log n})}$ because this is the best time complexity for deterministically computing a network decomposition of $G^2$ in the \CONGEST model \cite{DISC18_DomSet} (i.e., a network decomposition, where clusters of the same color are at distance at least $2$ from each other). An improved deterministic network decomposition algorithm would directly also lead to an improved deterministic MDS algorithm.\footnote{We note that in a recent paper, Ghaffari \cite{ghaffari19_MIS} showed that there is a deterministic \CONGEST model algorithm to compute a strong diameter decomposition of $G$ in time $2^{O(\sqrt{\log n})}$. In on-going unpublished work \cite{ghaffari19_personal}, it is even shown that there is a deterministic \CONGEST algorithm to compute a weak diameter decomposition of powers of $G$ in time $2^{O(\sqrt{\log n})}$, which has clusters that can be efficiently used in the \CONGEST model. As a result, the time complexity of the above theorem can be improved from $2^{O(\sqrt{\log n \log\log n})}$ to $2^{O(\sqrt{\log n})}$.}

In addition to the algorithm based on network decomposition, we also
provide a more local way to derandomize a similar randomized \CONGEST
algorithm, resulting in a time complexity that mainly depends on the
maximum degree $\Delta$ of the graph.

\begin{theorem} \label{thm:mainThmDelta} For any $\eps>1/\polylog \Delta$,
  there exists a deterministic \CONGEST algorithm that computes an
  $(1+\eps)\cdot (1+\ln (\Delta+1))$-approximation for MDS in $O(\Delta\cdot \polylog \Delta + \polylog\Delta\cdot \logstar n)$
  rounds. 
\end{theorem}
We want to emphasize that, except for the $O(\log^2 n)$-approximation algorithm in \cite{DISC18_DomSet}, our results are the first deterministic algorithms for MDS in the \CONGEST model.
By substituting a  vertex coloring subroutine in the algorithm of \Cref{thm:mainThmDelta} by its \LOCAL model counterpart this directly also leads to  an improved and slightly more efficient deterministic
distributed MDS algorithm in the \LOCAL model.

\begin{corollary} \label{corr:mainCorDelta} For any
  $\eps>1/\polylog \Delta$, there exists a deterministic \LOCAL
  algorithm that computes an $(1+\eps)\cdot \ln (\Delta+1)$ approximation
  for MDS in
  $O(\Delta\cdot \polylog \Delta + \logstar n)$ rounds.
\end{corollary}

In addition to considering the basic dominating set problem, we also
provide an algorithm to compute a small connected dominating set
(CDS). Here, the objective is to output a dominating set $S$ such that
the induced subgraph $G[S]$ is connected. It is well known that a
given dominating set $S$ can be extended to a CDS $S'$ of size at most
$|S'|< 3|S|$ by connecting the nodes in $S$ through at most $|S|-1$
paths of length at most $3$. One computes a spanning tree
on the graph induced by the set of nodes in $S$ and with an edge between two
nodes in $S$ whenever they are at distance at most $3$ in $G$. In the
\LOCAL model, one can get a local variant of this construction by
replacing the spanning tree with a high girth connected spanning
subgraph, which can be computed in logarithmic time.  As for
example observed by Dubhashi et al.~\cite{dubhashi05}, this leads to a
randomized polylog-round distributed algorithm with approximation
ratio $O(\log\Delta)$. Ghaffari \cite{ghaffari14_CDS} later showed that
by using an adapted variant of the linear-size skeleton algorithm of Pettie
\cite{pettie10}, it is also possible to compute an
$O(\log\Delta)$-approximation in polylogarithmic randomized time in
the \CONGEST model. We show that by using network decomposition, an
extended variant of Ghaffari's algorithm can be derandomized to obtain
the following result. 

\begin{theorem} \label{thm:mainThmCDS} There exists a deterministic
  \CONGEST algorithm that computes an
  $O(\ln\Delta)$ approximation for minimum connected dominating
  set in $2^{O(\sqrt{\log n \log \log n})}$ rounds. 
\end{theorem}
 
As for the basic MDS problem, an improved deterministic network decomposition
algorithm would directly also lead to an improved deterministic
algorithm for the minimum CDS problem.

\subsection{Our Dominating Set Algorithm in a Nutshell}
\label{sec:nutshell}
\paragraph{The general algorithmic framework:} Our algorithms start with a given fractional dominating set solution, which is $1/\poly\Delta$-fractional\footnote{In a \emph{fractional dominating set} each node $u$ is assigned a value $x_u\in[0,1]$ and each node $v\in V$ needs to be covered, that is, the sum $\sum_{u\in N(v)}x_u$ is at least $1$ where $N(v)$ denotes the inclusive neighborhood of $v$. We say that a fractional dominating set solution is $\lambda$-fractional if each non-zero fractional value is $\geq \lambda$. }  and can be computed efficiently in the \CONGEST model by using the algorithm of \cite{Kuhn06}. This fractional solution is then gradually and iteratively rounded to an integral solution. To explain a high-level view on the algorithm, let us fix some $\eps=1/\poly\log\Delta$.  The rounding consist of two parts. First, in $O(\log\Delta)$ phases, we roughly double the fractionality of the fractional solution in each phase until we obtain a $1/F$-fractional dominating set for some $F=\poly\log\Delta$. In each of these iterations, the dominating set might increase by a factor $1+O(\eps/\log \Delta)$ so that over all $O(\log \Delta)$ rounding steps, we only lose a $(1+O(\eps))$-factor.  Once we have a $1/F$-fractional dominating set, we round it into an integral dominating set with a single rounding step. During this step we essentially lose a $\ln (\Delta+1)$ factor in the approximation guarantee.  The rounding procedures in both parts are based on the derandomization of simple one-round randomized rounding primitives; derandomizing them in the \CONGEST model is one of the main technical contributions of this paper. 

\paragraph{Derandomization with network decomposition (runtime as a function of $n$):} In order to implement the approach based on a given network decomposition, we show that the above algorithm can also be implemented if the nodes only have $\polylog n$-wise independent bits. One can then share a seed of $\polylog n$ random bits in each cluster of the decomposition, which allows to create sufficiently many $\polylog n$-wise independent random bits. Finally, one can then use the method of conditional expectation inside each cluster to deterministically determine a seed such that the dominating set is small if we execute the rounding process  with that seed.

\paragraph{Derandomization with graph colorings (runtime as a function of $\Delta$):} In the original randomized rounding  process (with truly independent random bits), nodes at distance at least $3$ act independently. By applying the method of conditional expectation, one can derandomize the rounding algorithm as long as the actions of nodes at distance at most $2$ are carried out consecutively (see also \cite{newHypergraphMatching}). If a vertex coloring of the $G^2$ is given, this leads to an algorithm with a running time that is linear in the number of colors of the given $G^2$-coloring. This directly leads to an algorithm, with a running time of essentially $O(\Delta^2)$, the maximum degree of $G^2$. The main challenge in this part of the paper is to substitute one of the $\Delta$ factors by a $\polylog \Delta$ factor. The main idea  is to substitute $G$ with its \emph{bipartite representation} $B_G$, where each node is split into a \emph{constraint node} and a \emph{value node} (see \Cref{sec:derandomizationD} for more details). To reduce the number of colors we carefully split the constraint nodes of $B_G$ into nodes with much smaller degrees such that it suffices to compute a $O(\Delta\polylog \Delta)$-coloring of the square of this bipartite graph.

\paragraph{Why do we not just have one iteration of rounding?} Our algorithms first double the fractionality for $O(\log \Delta)$ iterations and in the last step, if the fractionality is at least $1/\polylog \Delta$ rounds to an integral solution in a single iteration. Interestingly, this iterative rounding approach is helpful in both cases, however, for different reasons: When we express the runtime as a function of $n$ we need to round in small steps to make sure that $\polylog n$-wise independent random bits are sufficient. When we express runtime as a function of $\Delta$, we can only split the constraint nodes of the bipartite representation of the problem into small enough parts and still obtain suffcient concentration in the randomized rounding if we round gradually.

\subsection{Outline}
In \Cref{sec:model} we formally introduce the model, fractional dominating sets and recall that the latter can be computed efficiently with the algorithm from \cite{Kuhn06}. \Cref{sec:domSetapprox} is the technically most involved part of the paper and it is split into four subsections. In \Cref{sec:abstractRounding} we introduce the abstract randomized rounding primitive which takes a fractional dominating set and a "rounding probability" for each node as input. Then, the randomized algorithm "rounds" the fractional values of the dominating set with these probabilities to increase the fractionality of the dominating set. Our main algorithm that is presented in \Cref{sec:proofs} uses instantiations of the rounding primitive to round a fractional dominating set into an integral dominating set.
 However, to ensure that our main algorithm is deterministic we show in \Cref{sec:derandomizationN,sec:derandomizationD} that, if the probabilities are chosen in the right way, we can derandomize this simple randomized algorithm efficiently. While  \Cref{sec:derandomizationN} tries to optimize the runtime of the derandomization process as a function of the number of nodes in the network, \Cref{sec:derandomizationD} tries to optimize the runtime as a function of the maximum degree. Due to the different objective both sections greatly differ in how the abstract rounding primitive is derandomized. 
In \Cref{sec:CDS} we show how the computed dominating sets can be transformed into connected dominating sets.
In \Cref{sec:conclusions} discuss how our results can be generalized.

\section{Model \& Notation}
\label{sec:model}
\paragraph{The \LOCAL Model of Distributed Computing~\cite{linial92, peleg2000distributed}.} The graph is abstracted as an $n$-node network $G=(V, E)$ with maximum degree at most $\Delta$ and each node has a unique identifier. Communications happen in synchronous rounds. Per round, each node can send one (unbounded size) message to each of its neighbors. At the end, each node should know its own part of the output, e.g., whether it belongs to a dominating set or not.
\paragraph{The \CONGEST Model~\cite{peleg2000distributed}.} In almost all parts of the paper we use the more restricted \CONGEST model which is defined like the \LOCAL model with the restriction that messages are of size $O(\log n)$. So messages can fit, e.g., a constant number of node identifiers. The most complicated computations that nodes need to perform in our algorithms are the computations of certain conditional probabilities. 

We say that a value in $[0,1]$ is \CONGEST \emph{transmittable} if it is the multiple of $2^{-\iota}$ where $\iota$ is the smallest integer such that $2^{-\iota}\leq 1/n^{10}$. During our algorithms we merely require that a value fits into a single message in the \CONGEST model and additionally we require that a biased coin that equals $1$ with a transmittable probability can be built with polylogarithmically many fair coins. 
\paragraph{Graph Notation.}
Given a graph $G=(V,E)$ we use $N(v)$ to denote the inclusive neighborhood of a node $v\in V$, $\Deltab=\max_{v\in V}|N(v)|=\Delta+1$ for the size of the inclusive neighborhood and $\Delta$ for the maximum degree of the graph.
We denote the shortest hop distance between two nodes $u,v\in V$ as $d_G(u,v)$. To simplify calculations we assume that $\eps\leq 1$ throughout all results.
For our algorithms we need to generalize the notion of a dominating set.

\begin{definition}[Constraint Fractional Dominating Set]
  Given a graph $G=(V,E)$ and a function $c: V \rightarrow [0,1]$ assigning constraints and a function $x: V \rightarrow [0,1]$ assigning fractional values, we call $(x,c)$ a \emph{constrained fractional dominating (CFDS)} set if $\forall v \in V : \sum_{u \in N(v)}x(u) \geq c(v)$. We furthermore call $\sum_{v\in V} x(v)$ the size of the CFDS. If all constraints equal $1$ for all nodes and $(x,c)$ is a constrained fractional dominating set we also call $x$ a \emph{fractional dominating set (FDS)}. 
\end{definition}
If the fractional values equal $1$ or $0$ for all nodes and $x$ is a FDS then $x$ is a dominating set (DS) in the classical sense. 
All of our algorithms start with a fractional dominating set that already is a good approximation but has a bad fractionality. Computing such an initial fractional dominating set can be done by slightly adapting an algorithm from \cite{Kuhn06}. 
\begin{restatable}[\cite{Kuhn06}]{lemma}{lemFractionalApproximation}
\label{lem:fractionalApproximation}
For any $\eps> 0$ there is a deterministic \CONGEST algorithm that computes a $(1+\eps)$-approximation for MDS that is $\eps/(2\Delta)$-fractional and has runtime $O(\eps^{-4}\log^2\Delta)$.
\end{restatable}
\begin{proof}
First use the algorithm from \cite{Kuhn06} with $\eps'=\eps/2$ to compute a $(1+\eps)$-approximation for MDS. Then, each node with value $<\eps/(2\Delta)$ sets its value to $\eps/(2\Delta)$. As the size of an optimal dominating can be lower bounded by  $n/\Delta$  we obtain a $(1+\eps)$ approximation.
\end{proof}

\section{Dominating Set Approximation}
\label{sec:domSetapprox}
In this section we provide two deterministic algorithms to approximate the minimum dominating set problem in the \CONGEST model. The runtime of the first algorithm (cf. \Cref{thm:mainThmN}) depends on the number of nodes $n$ of the graph and the runtime of the second algorithm (cf. \Cref{thm:mainThmDelta}) mostly depends on the maximum degree $\Delta$ of the graph. Both algorithms rely on deterministic rounding schemes that are obtained through derandomizing instances of a more general abstract randomized rounding process. We use \Cref{sec:abstractRounding} to define this abstract rounding process, \Cref{sec:derandomizationN} to derandomize the process while optimizing dependency on $n$ and \Cref{sec:derandomizationD} to derandomize the process while optimizing the dependency on $\Delta$. In \Cref{sec:proofs} the rounding schemes are combined to prove our main theorems.

\subsection{Abstract Randomized Rounding}
\label{sec:abstractRounding}
The goal of the abstract randomized rounding process is to transform a constraint fractional dominating set into a new CFDS with better fractionality while only slightly increasing its size.

\paragraph{Abstract Randomized Rounding Process:}
Assume a given (initial) constraint fractional  dominating set with values $x(v), v\in V$. 
Further, each node has a probability $p(v)\geq x(v)$ that might depend on $x(v)$. The rounding process is parameterized by the choice of $p(v)$---we will later essentially consider two variants of the process with different $p(v)$'s.
 The rounding process consists of two phases and the random variable $X_v$ denotes $v$'s value after the first phase.  
 In the first phase each node $v$ updates its value as follows
\begin{align*}
X_v = \left\{\begin{array}{lr}
        x(v)/p(v), & \text{with probability $p(v)$}\\
        0, & \text{with probability $1-p(v)$}. 
        \end{array}\right. 
\end{align*}

In the second phase of the rounding process all nodes whose constraints are not satisfied join the dominating set, i.e., they set their values to $1$.

\medskip 

For node $v$ let $E_v$ be the random variable that equals $1$ if node $v$ is uncovered after the first phase of the algorithm and $0$ otherwise.

\begin{lemma} \label{lem:abstractRoundingWithConstraints}

\label{lem:abstractRounding}
Let $A$ be an upper bound on the size of the initial constraint fractional dominating set. Then the following statements hold for the abstract randomized rounding process.
\begin{enumerate}
\item Its output is a constraint fractional dominating set and its fractionality is $\min_{v\in V}\{x(v)/p(v)\}$.
\item Its output has, in expectation, size $A+\sum_{v\in V}\Pr(E_v=1)$.
\end{enumerate}
\end{lemma}
\begin{proof}
\begin{enumerate}
  \item[]
\item The output is a constraint fractional dominating set as all nodes that are uncovered after phase one join the dominating set with value $1$ in phase two and constraints can be at most $1$. As each node either sets its value to $0$, to $x(v)/p(v)$ or to $1$ the fractionality can be lower bounded by  $\min_{v\in V}\{x(v)/p(v)\}$.
\item Let $Z_1$ be the size of the dominating set after the first phase and let $Z_2$ be the number of nodes that join the dominating set in phase two. The dominating set size $Z$ after phase two is upper bounded by $Z_1+Z_2$ and we have 
\end{enumerate}
\begin{align*}
  \E[Z]& \leq \E[Z_1]+\E[Z_2]=\sum_{v\in V}p(v)\cdot x(v)/p(v) +  \sum_{v\in V}~\Pr(E_v=1) \\
	& =A+ \sum_{v\in V}~\Pr(E_v=1)~. 
\end{align*}
\end{proof}


\subsection{Derandomization via Network Decompositions}
\label{sec:derandomizationN}
In this section we  derandomize instantiations of the aforementioned abstract randomized rounding process with the help of $k$-hop network decompositions. \cite{DISC18_DomSet} shows how these can be computed in the \CONGEST model and we restate their definitions and results.
 We begin with the notion of a cluster graph.
\begin{definition}[Cluster Graph]Given a graph $G=(V,E)$ and an integer $d\geq 1$ an $(N,d)$-cluster graph $\mathcal{G}=(\mathcal{V},\mathcal{E})$ is a graph that is given by a set of clusters $\mathcal{V}=\{\CC_1,\ldots,\CC_N\}\in 2^V$ such that 
\begin{enumerate}
\item  the clusters form a partition of $V$, 
\item each cluster $\mathcal{C}_i$ induces a connected subgraph $G[\CC_i]$ of $G$,
\item each cluster $\CC_i$ has a leader $\ell(\CC_i)$ that is known by all nodes of $\CC_i$,
\item  inside each cluster, there is a rooted spanning tree $T(\CC_i)$ of $G[\CC_i]$ that is rooted at $\ell(\CC_i)$ and has diameter at most $d$. 
\end{enumerate}
There is
an edge $\{\CC_i,\CC_j\}$ between two clusters $\CC_i,\CC_j \in\mathcal{V}$ if there is edge in $G$ connecting a node in
$\CC_i$ to a node in $\CC_j$ . The identifier $ID(\CC_i)$ of a cluster $\CC_i$ is its leader’s ID.
\end{definition}
Given a cluster graph $\mathcal{G} = (\mathcal{V}, \mathcal{E})$ of $G$ and an integer $k \geq 1$, we say that two clusters
$\CC,\CC'\in \mathcal{V}$ are $k$-separated if for any two nodes $u$ and $v$ of $G$ such that $u\in \CC$ and  $v \in \CC'$,
we have $d_G(u, v) > k$. A strong-diameter $k$-hop network decomposition of a graph $G$ is then
defined as follows.
\begin{definition}[Network Decomposition] Let $G = (V,E)$ be a graph and let $k \geq 1, d \geq 0$,
and $c\geq 1$ be integer parameters. A strong diameter $k$-hop $(d, c)$-decomposition of $G$ is a
$(N, d)$-cluster graph $\mathcal{G}$ of $G$ for some integer $N \geq 1$ together with a coloring of the clusters of
$G$ with colors $\{1, . . . , c\}$ such that any two clusters with the same color are $k$-separated.
\end{definition}
\cite{DISC18_DomSet} shows how to compute $k$-hop network decompositions in the \CONGEST model.
\begin{theorem}[Theorem 3 in \cite{DISC18_DomSet}]
\label{thm:networkDecomp}
Let $G = (V,E)$ be an $n$-node graph and let $k \geq 1$ be an integer. There
is a deterministic \CONGEST model algorithm that computes a strong diameter $k$-hop $(k \cdot
f(n), f(n))$-decomposition of G in $k \cdot f(n)$ rounds, where $f(n) = 2^{O(\sqrt{\log n\cdot\log \log n})}$.
\end{theorem}

In the proof of our main derandomization lemma we define a process that only needs $\polylog n$-wise independent random bits. 
A finite set $\mathcal{X}$ of biased coins is \emph{$k$-wise independent} if any subset $T\subseteq \mathcal{X}$ of size $k$ is independent in the classical sense. We use the following classic result, which is for example proven in \cite{alon2004probabilistic}.

\begin{lemma}\label{lem:randomSeedCreation}
Let $N$ and $k$ be integers and $p_1,\ldots,p_N\in[0,1]$ multiplies of $2^{-\iota}$ for some $\iota=O(\log N)$. Given a random seed consisting of $K=O(k\log^2 N)$ independent fair coins, one can efficiently and deterministically extract $N$ biased coins with probabilities $p_1,\ldots,p_N$ from the seed and these coins are $k$-wise independent. 
\end{lemma}
We can now state our main derandomization lemma with the help of network decompositions.
\begin{restatable}[Derandomization Lemma I]{lemma}{lemDerandomization}
\label{lem:derandomization}
There is a determ. \CONGEST algorithm that, given 
\begin{itemize}
\item a constraint fractional dominating set of size $A$ with transmittable values $x(v), v\in V$,
\item a transmittable probability $p(v)\geq x(v)$ for each $v\in V$ and
\item a sufficiently large parameter $k=\polylog n$,
\end{itemize} computes a constraint fractional dominating set with fractionality $\min_{v\in V}\{ x(v)/p(v)\}$ and size at most $A+\sum_{v\in V}\Pr(E_v=1) +\frac{1}{n^6}$ where $E_v$ is the event that $v$ is uncovered after the first phase if the abstract randomized rounding process is executed with $x(v)$ and $p(v)$ for $v\in V$ and $k$-wise independent (biased) coins for the probabilities $p(v)$. 
 The round complexity of the algorithm is  $2^{O(\sqrt{\log n \log \log n})}$~. 
\end{restatable}
The proof of \Cref{lem:derandomization} is based on iterating through the clusters of a network decomposition (of $G^2$) and creating a random seed of $\polylog n$ length inside each cluster. This is sufficient to create $\polylog n$-wise independent bits for all nodes to execute the rounding algorithm. Then the algorithm is derandomized, i.e., bits of the random seeds are fixed deterministically with the method of conditional expectation.

\begin{proof}
We show that we can derandomize the aforementioned abstract randomized rounding algorithm in the \CONGEST model in $2^{O(\sqrt{\log n \log \log n})}$ rounds. In particular we derandomize the following \emph{'almost equivalent'} randomized process:
 Deterministically compute a $2$-hop $(d,c)$-network decomposition  with $d,c=2^{O(\sqrt{\log n \log \log n})}$ in $2^{O(\sqrt{\log n \log \log n})}$ rounds (cf. \Cref{thm:networkDecomp}). Then iterate through the $c$ color classes. In iteration $i$ each leader of all clusters with color $i$ creates a random seed of length $K$---the randomness for distinct clusters is independent---and distributes the seed to all members of the cluster. $K=\polylog n$ is chosen as in \Cref{lem:randomSeedCreation} such that each node $v$ of the cluster can extract a biased coin (from the random seed)  that equals one with probability $p(v)$ and zero otherwise and such that these coin flips are $k$--wise independent.  After all clusters have been processed all nodes execute both phases of the abstract randomized rounding algorithm with their respective random bits. Note that the only randomness used by the algorithm is the generation of the random seeds---one random seed, that is, $K$ random bits, per cluster. 
Due to linearity of expectation all results of \Cref{lem:abstractRounding} hold despite the limited dependence between coins and thus the process computes a fractional dominating set with fractionality $\min\{x(v)/p(v)\}$ and expected size $A+\sum_{v\in V}\Pr(E_v=1)$.

To derandomize the above algorithm we use the method of conditional expectation, that is, we fix the random bits in the above algorithm one after another where the expectation (over the randomness of all unset random bits) only increases marginally during each step. There are $K$ random bits per cluster and fixing a single bit takes $O(d)$ rounds---recall that $d$ is the maximum cluster diameter. We will show that random bits of distinct clusters in the same color class can be fixed simultaneously. The full deterministic algorithm has the following steps: 1) compute a $2$-hop $(d,c)$-network decomposition, 2) deterministically determine a random seed for each cluster, 3) deterministically execute the 'randomized algorithm' with these random seeds (using the same network decomposition).
The first step needs $2^{O(\sqrt{\log n \log \log n})}$ rounds, the second steps  $O(K\cdot c\cdot d)=2^{O(\sqrt{\log n \log \log n})}$ rounds and the third step can be executed in $O(1)$ rounds.

We now show how the random bits of the random seeds can be fixed.  Assume that clusters are ordered from low color to high color (and arbitrarily within the same color classes). Let $B_{i,j}$ denote the $j^{th}$ bit of the random seed of cluster $i$. We explain how to fix random bits in the order $B_{1,1},B_{1,2},\ldots, B_{1,K},B_{2,1},\ldots,B_{2,K},\ldots$, however it will turn out that random bits of distinct clusters of the same color can be fixed at the same time.
To reduce notation assume that bits are indexed with $1,2,3,\ldots$ and assume that we want to deterministically determine the $j^{th}$ bit. Let  $\mathcal{C}$ be the cluster of the $j^{th}$ bit and let $\ell$ be the cluster leader.  With constant overhead we can assume that each node $v$ always knows its neighbors'  $x(u)$ and $p_u$ values as well as its neighbors' IDs and the algorithm will be designed such that $v$ also knows all fixed random bits that influence the execution of the first phase of any node in $N(v)$.  Let bits $B_1,\ldots, B_{j-1}$ be already fixed as $b_1,\ldots,b_{j-1}$.
With this knowledge each node $v\in N(\mathcal{C})$ computes the following two values where the random variable $Z_v$ denotes the value of $v$ after the second phase of the process.\footnote{Here $N(\mathcal{C})$ denotes the nodes inside the cluster and nodes that have at least one neighbor in $\mathcal{C}$ in graph $G$.}
\begin{align*}
\alpha_{v,0} & =\E[Z_v \mid B_0=b_0, \ldots, B_{j-1}=b_{j-1}, B_j=0] \\ 
\alpha_{v,1} & =\E[Z_v  \mid B_0=b_0, \ldots, B_{j-1}=b_{j-1}, B_j=1]
\end{align*} 
Then node $v$ rounds these values with accuracy $1/n^{10}$ and obtains values $\tilde{\alpha}_{v,0}$ and $\tilde{\alpha}_{v,1}$, respectively.  Then, as the values are rounded to multiples of $1/n^{10}$, we can aggregate their respective sums at $\ell$ in $O(d)$ rounds using the spanning tree of the cluster; here, nodes that have a neighbor in the cluster but are not part of the cluster send their value to one neighbor in the cluster as if it was the parent in the spanning tree. 
Thus $\ell$ learns the two values $\sum_{v\in N(\mathcal{C})}\tilde{\alpha}_{v,0}$ and $\sum_{v\in N(\mathcal{C})}\tilde{\alpha}_{v,1}$. If the first one is smaller $\ell$ sets $b_j=0$ and otherwise it sets $b_j=1$. Then the decision is send to all nodes in $N(\mathcal{C})$ and we continue with the next bit. We now prove the following statements:

\noindent\textbf{Claim:} The expected size of the FDS increases at most by $2/n^9$ by fixing a single bit.
\renewcommand\qedsymbol{$\blacksquare$}
\begin{proof} To succinctly express conditional expectations with already fixed bits introduce the notation
$\theta=\bigwedge_{l=1}^{j-1}B_l=b_l$ . 
Formally we have
\begin{align*}
b_j& =\argmin_{x\in\{0,1\}}\left\{\E\left[\sum_{v\in N(\mathcal{C})}\tilde{\alpha}_{v,x} \mid \theta, B_j=x\right]\right\}
\end{align*}
Let the random variable $Z$ be the size of the dominating set after the second phase. We can upper bound the expected influence of nodes in $N(\mathcal{C})$ on the size of the dominating set after fixing the $j^{th}$ bit:
\begin{align*}
  \mathbb{E}[Z \mid \theta \land B_j = b_j] - \mathbb{E}[Z \mid \theta]
  &= \sum_{v\in V} \left(\mathbb{E}[Z_v \mid \theta \land B_j = b_j] - \mathbb{E}[Z_v \mid \theta]\right)\\
  &\stackrel{(*)}{=} \sum_{v\in N(\mathcal{C})} \left(\mathbb{E}[Z_v \mid \theta \land B_j = b_j] - \mathbb{E}[Z_v \mid \theta]\right)\\
  &= \sum_{v\in N(\mathcal{C})} \alpha_{v,b_j} - \mathbb{E}\left[\sum_{v\in N(\mathcal{C})} Z_v \mid \theta \right]\\
  &\leq \sum_{v\in N(\mathcal{C})} \left(\tilde{\alpha}_{v,b_j} + \frac{1}{n^{10}}\right) - \mathbb{E}\left[\sum_{v\in N(\mathcal{C})}Z_v \mid \theta\right]\\
  &\leq  \frac{1}{n^9} + \sum_{v\in N(\mathcal{C})} \tilde{\alpha}_{v,b_j} - \mathbb{E}\left[\sum_{v\in N(\mathcal{C})}Z_v \mid \theta\right]\\
&\stackrel{(**)}{=}  \frac{1}{n^9} +\min\left\{ \sum_{v\in N(\mathcal{C})} \tilde{\alpha}_{v,0},\sum_{v\in N(\mathcal{C})} \tilde{\alpha}_{v,1}\right\} - \mathbb{E}\left[\sum_{v\in N(\mathcal{C})}Z_v \mid \theta\right]\\
&\leq  \frac{2}{n^9} +\min\left\{ \sum_{v\in N(\mathcal{C})} \alpha_{v,0},\sum_{v\in N(\mathcal{C})} \alpha_{v,1}\right\} - \mathbb{E}\left[\sum_{v\in N(\mathcal{C})}Z_v \mid \theta\right]\\
&\leq  \frac{2}{n^9} +\sum_{v\in N(\mathcal{C})} \left(\min\left\{ \alpha_{v,0}, \alpha_{v,1}\right\} - \mathbb{E}\left[Z_v \mid \theta\right]\right)\\
&\stackrel{(***)}{\leq}  \frac{2}{n^9} \leq \frac{1}{n^8}~,
\end{align*}
where $(*)$ follows because $B_j$ does not influence $Z_v$ for $v\notin N(\mathcal{C})$, 
 $(**)$ follows due to the choice of $b_j$ and $(***)$ follows as $\min\left\{ \alpha_{v,0}, \alpha_{v,1}\right\} - \mathbb{E}\left[Z_v \mid \theta\right]\leq 0$ holds due to the law of total expectation.
\end{proof}
\renewcommand\qedsymbol{$\square$}

There are at most $n$ clusters and $K$ bits per cluster, thus at the end the size of the dominating set is upper bounded by $A+\sum_{v\in V}\Pr(E_v=1)+\frac{n\cdot K\cdot 2}{n^8}\leq A+\sum_{v\in V}\Pr(E_v=1)+n^{-6}$ for sufficiently large $n$.

\noindent\textbf{Claim:} Bits of distinct clusters with the same color can be fixed at the same time.
\renewcommand\qedsymbol{$\blacksquare$}
\begin{proof}
Let $\mathcal{C}_1$ and $\mathcal{C}_2$ be two distinct clusters of the same color. The computed network decomposition is a $2$-hop decomposition. Hence the distance between the two clusters in $G$ is at least three and $N(\mathcal{C}_1)$ and $N(\mathcal{C}_2)$ are disjoint. Thus the set of random variables $\{Z_v \mid v\in N(\mathcal{C}_1)\}$ and $\{Z_v \mid v\in N(\mathcal{C}_2)\}$, respectively, that are influenced by the bits of the random seed in cluster $\mathcal{C}_1$ and $\mathcal{C}_2$, respectively, are disjoint. 
\end{proof}
\renewcommand\qedsymbol{$\square$}
The fractionality of the resulting dominating set follows from the fractionality of the abstract randomized rounding algorithm.
\end{proof}

We define two algorithms that use the above rounding process (with different choices for $p(v)$) as their main subroutine. Both begin with an initial fractional dominating set with values $x'(v), v\in V$.
\paragraph{One shot rounding}
Increase the values of the input fractional dominating set by a factor $\ln\Deltab$ and obtain a fractional dominating set where node $v$ has value $x(v)=\min\{1,x'(v)\cdot \ln \Deltab\}$. Then execute the randomized rounding process  with value $x(v)$ and $p(v)=x(v)$ for all $v\in V$. 

Note that the process one shot rounding transforms a fractional dominating set into an integral dominating set. 
\paragraph{Factor two rounding}
Let $\eps>0$ and $r=\polylog n$. 
Increase the values of the input fractional dominating set by a factor $(1+\eps)$ and obtain a fractional dominating set where node $v$ has value $x(v)=\min\{1,x'(v)\cdot (1+\eps)\}$.
Then execute the randomized rounding process with values $x(v)$ for $v\in V$ and 
\begin{align*}
p(v) = \left\{\begin{array}{lr}
        \frac{1}{2}, & \text{if $x(v)<2/r$}\\
        1, & \text{if $x(v)\geq 2/r$}~. 
        \end{array}\right. 
\end{align*}
Thus in the factor two rounding nodes with value $x(v)<2/r$  double their value with probability $1/2$ and otherwise set their value to $0$. Nodes with value $x(v)\geq 2/r$ simply keep their value.

Recall that $E_v$ equals $1$ if node $v$ is uncovered after the first phase and $0$ otherwise. The proofs of the following  two crucial lemmas upper bound the probability of $\Pr(E_v=1)$ and use the following Chernoff type bound for $k$-wise independent random varibales.
\begin{theorem}[Theorem 5 in \cite{chernoffLimited95}] \label{thm:kindependent}
Let $X$ be the sum of $k$-wise independent $[0,\lambda]$-valued random variables with expectation $\mu=E(X)$ and let $\delta\leq 1$. Then we have
\[Pr(|X-\mu|\geq \delta\mu)\leq e^{-\lfloor \min\{k/2,\delta^2\mu/(3\lambda)\}\rfloor}~.\] 
\end{theorem}

\begin{restatable}{lemma}{lemOneShotUncover}
  \label{lem:oneshotUncover}
  Assume a given $1/F$-fractional dominating set for some integer $F>0$. Then in the one shot rounding process we have $Pr(E_v=1)\leq \Deltab^{-1}$. The result holds even if the biased coins of the nodes are only $k$-wise independent for any $k\geq F$.
\end{restatable}
\begin{proof}
  Let $v\in V$ be a node. As the initial dominating set is $1/F$-fractional $v$ is already covered by $F$ nodes of $N(v)$. Let $S\subseteq N(v)$ denote a set of $F$ nodes that cover $v$. If any $u\in S$ has $x(u)=1$ node $v$ is covered. Assume that $x(u)<1$ for all $u\in S$. Then we have  $x(u)= \ln \Delta\cdot x'(u)$ for all $u\in S$. Let $Y=\sum_{u\in S}X_u$ and let 
\begin{align*}
  \mu:=\E[Y]=\sum_{u\in S}p(u)\cdot \frac{x(u)}{p(u)}+(1-p(u))\cdot 0 \geq \ln\Deltab \sum_{u\in S}x'(u)\geq \ln \Deltab~.
\end{align*}
$Y$ is the sum of $[0,1]$-valued random variables that are independent as $k\geq F\geq |S|$. We obtain
\begin{align*}
\Pr(v \text{ is uncovered}) & \leq \prod_{u\in S}{(1-\Pr(X_u=1))}  
  \leq  e^{-\ln\Deltab \cdot \sum_{u\in S}x'(v)} 
 \leq e^{-\ln \Deltab} = \Deltab^{-1} ~. & & \qedhere
\end{align*}
\end{proof}

\begin{restatable}{lemma}{lemFactorTwoUncover}
\label{lem:factorTwoUncover}
Let $\eps>0$, $r\geq 256\eps^{-3}\ln \Deltab$ and assume a given fractional dominating set.
In the factor two rounding process we have $Pr(E_v=1)\leq \Deltab^{-4}$.
The result holds even if the biased coins of the nodes are only $k$-wise independent for some $k\geq 8\ln \Deltab$.
\end{restatable}

\begin{proof}
Recall that $x'(v)$ is the value of node $v$ before the process is executed, $x(v)$ is the value before the randomized rounding is executed and $X_v$ is its value after the first phase. $E_v$ equals $1$, if node $v$ is uncovered after the first phase, that is, if $\sum_{u\in N(v)}X_v<1$.

We say, that a node with $p(v)<1$ \emph{takes part} in the rounding process and a node with $p(v)=1$ \emph{does not take part} in the rounding process. Let $F=r$. 
Fix a node $v\in V$ and let $\NF=\{u\in N^+(v) \mid x(u)\leq 1/F\}$ be the set of its neighbors that take part in the rounding process. Further, let $\NNF=N(v)\setminus \NF$ be the nodes that do not take part in the rounding process. Note that all $u\in \NNF$ have $x(u)\geq 1/F$ and thus either $X_u$ equals $1$ or $X_u\geq (1+\eps)x'(u)$. If $X_u=1$ for some $u\in \NNF$ node $v$ is covered, so without loss of generality we may assume that $X_u\geq (1+\eps)x'(u)$ for all $u\in \NNF$.

Let $S_v= \sum_{u\in \NF} x'(u)$ be the total initial value of the neighbors that take part in the rounding and  let $T_v=\sum_{u\in \NNF} x'(u)$ be the total initial value of the neighbors that do not take part in the rounding. As the initial values form a dominating set we have $S_v+T_v\geq 1$. We perform a case distinction depending on the value of $S_v$. 

\noindent\textit{Case $S_v\leq \frac{\eps}{1+\eps}$:}
If there is a node $u\in \NNF$ with $X_u=1$ node $v$ is covered. Otherwise all nodes nodes $u\in \NNF$ have $X_u\geq (1+\eps)x'(u)$ and we have 
\[\sum_{u\in  \NNF} X_u\geq\sum_{u\in  \NNF} (1+\eps)x'(u)\geq (1+\eps)T_v\geq (1+\eps)\cdot (1-S_v)\geq 1~.\] 
Thus in this case $v$ is covered by the nodes in $\NNF$. 

\medskip 

\noindent\textit{Case $S_v>\frac{\eps}{1+\eps}$:}
 If $v$ is uncovered we have
\begin{align*}
1 & >\sum_{u\in N(v)}X_u=\sum_{u\in  \NNF} X_u+\sum_{u\in  \NF} X_u\\
& \stackrel{(*)}{\geq} (1+\eps)\cdot (1-S_v)+ \sum_{u\in  \NF} X_u\geq  (1-S_v)+\sum_{u\in  \NF} X_u
\end{align*}
where $(*)$ follows with the same argumentation as in the first case if $v$ is not covered by a single node in $\NNF$.
Thus if $v$ is uncovered we have  $\sum_{u\in  \NF} X_u<S_v$. We next bound the probability of this to happen which yields the claim. Denote $Y=\sum_{u\in  \NF} X_u$ and node that $Y$ is the sum of $k$-wise independent $[0,2(1+\eps)/F]$ valued random variables as $X_u\leq x(u)/p_u=2x(u)\leq 2(1+\eps) x'(u)\leq 2(1+\eps)1/F$ for $u\in \NF$. We compute $Y$'s expectation as
\begin{align*}
\mu:=\E\left[Y\right]=\sum_{u\in  \NF}p_u\cdot \frac{x(u)}{p_u}+0\cdot (1-p(v))=\sum_{u\in  \NF}(1+\eps)x'(u)=(1+\eps)S_v~.
\end{align*}
By \Cref{thm:kindependent} we have
\begin{align*}
\Pr\left(Y<S_v\right) & =\Pr(Y-\mu<S_v-\mu) 
 \leq  \Pr(|Y-\mu|\geq \eps S_v)=\Pr\left(|Y-\mu|\geq \frac{\eps}{1+\eps} \mu\right) \\ 
&  \leq e^{-\min\{k/2,\lceil\left(\frac{\eps}{1+\eps}\right)^2\mu/(3\cdot 2(1+\eps))\cdot F\rceil\}} 
 \stackrel{(*)}{\leq} e^{-\min\{k/2,\eps^3/48\cdot F\}}\stackrel{(**)}{\leq}  \frac{1}{\Deltab^4}
\end{align*}
where $(*)$ holds as $\mu\geq \eps$ and $\eps\leq 1$ and $(**)$ holds due to $k\geq 8 \log \Deltab $ and  $F\geq 256\eps^{-3} \log \Deltab $.\qedhere
\end{proof}

Note that \Cref{lem:oneshotUncover} works for all $F$, but if the fractionality is small we need larger independence. The derandomization lemma (\Cref{lem:derandomization}) is only efficient if we only need $k$-wise independence for some $k=\polylog n$. \Cref{lem:factorTwoUncover} crucially depends on a small fractionality as, otherwise if only few nodes double their value, there is no concentration around the expectation.

Applying \Cref{lem:derandomization} to both processes implies deterministic counterparts whose properties we summarize in the following two lemmas. 
\begin{restatable}[Deterministic One Shot Rounding]{lemma}{lemOneShotRounding}
\label{lem:oneShotRounding}
Let $F=\polylog n$. There is a deterministic \CONGEST algorithm that, given a $1/F$-fractional dominating set of size $A$, computes an integral dominating set with size at most $(\ln\Deltab)\cdot A+\frac{n}{\Deltab}+\frac{1}{n^5}$ and has round complexity $2^{O(\sqrt{\log n \log \log n})}$.
\end{restatable}

\begin{proof}
Apply \Cref{lem:derandomization} to the one shot rounding process with $k=F=\polylog n$ yielding an integral dominating set of size $(\ln\Deltab)\cdot A+\sum_{v\in V}\Pr(E_v=1)+n^{-6}+n^{-9}\leq (\ln\Deltab)\cdot A+\sum_{v\in V}\Pr(E_v=1)+n^{-5}$ where the $n^{-9}$ occurs as we have to increase each value such that it is transmittable before applying \Cref{lem:derandomization}.
\Cref{lem:oneshotUncover} bounds $\Pr(E_v=1)\leq \Deltab^{-1}$ for all $v\in V$ which proves the claim.
\end{proof}

\begin{restatable}[Deterministic Factor Two Rounding]{lemma}{lemFactorTwoRounding}
\label{lem:factorTwoRounding}
Let $\eps>1/\polylog n$, $r\geq F= 256\eps^{-3}\ln \Deltab$. There is a deterministic \CONGEST algorithm that, given a $1/r$-fractional dominating set with size $A$, computes a $2/r$-fractional dominating set with size $(1+\eps)A+\frac{n}{\Deltab^4}+\frac{1}{n^5}$ and has round complexity $2^{O(\sqrt{\log n \log \log n})}$.
\end{restatable}
\begin{proof}
Apply \Cref{lem:derandomization} to the one shot rounding process with $k=F=\polylog n$ yielding an integral dominating set of size $(\ln\Deltab)\cdot A+\sum_{v\in V}\Pr(E_v=1)+n^{-6}+n^{-9}\leq (\ln\Deltab)\cdot A+\sum_{v\in V}\Pr(E_v=1)+n^{-5}$ where the $n^{-9}$ occurs as we have to increase each value such that it is transmittable before applying \Cref{lem:derandomization}.
\Cref{lem:oneshotUncover} bounds $\Pr(E_v=1)\leq \Deltab^{-1}$ for all $v\in V$ which proves the claim.
\end{proof}



\subsection{Derandomization via Colorings}
\label{sec:derandomizationD}
In this subsection we present our second derandomization lemma which mainly expressed the runtime as a function of the maximum degree $\Delta$. In its proof we iterate through the color classes of a distance two coloring of a suitable set $S\subseteq V$ of the nodes of a given graph $G=(V,E)$ and at each node we use the method of conditional expectation to locally determine random bits that are sufficient to execute the abstract randomized rounding algorithm such that the computed dominating set is small enough. 
Here, a coloring of the nodes in $S$ is a \emph{distance two coloring} if $d_G(u,v)>2$ for any two nodes $u,v\in S$ that have the same color.

\begin{restatable}[Derandomization Lemma II]{lemma}{lemDerandomizationDelta}
\label{lem:derandomizationDelta}
There is a determ. \CONGEST algorithm that, given 
\begin{itemize}
  \item a constraint fractional dominating set $(x,c)$ of size $A$ with transmittable values $x(v), v\in V$,
\item a distance two $C$-coloring $\phi$ of $S=\{v\in V \mid p(v)\notin \{0,1\}\}$ and  
\item transmittable probabilities  $p(v)\geq x(v), v\in V$,
\end{itemize}
computes a constraint fractional dominating set with fractionality $\min_{v\in V}\{ x(v)/p(v)\}$ and size at most $A+\sum_{v\in V}\Pr(E_v=1) +\frac{1}{n^7}$ where $E_v$ is the event that $v\in V$ is uncovered after the first phase if the abstract randomized rounding process was executed with $x(v)$ and $p(v)$ for $v\in V$. The round complexity of the algorithm is  $O(C)$.
\end{restatable}

\begin{proof}
We derandomize the first phase of the abstract randomized rounding process with the method of conditional expectation where the objective function is the size of the dominating set after the second phase. For a node $v\in V$ let $X_v$ denote $v$'s value after the first phase of the process. Due to \Cref{lem:abstractRoundingWithConstraints}  the process computes a constraint fractional dominating set with the desired fractionality and its expected outcome is of size $A + \sum_{v \in V} \Pr(E_v = 1)$. Recall, that all nodes with $p(v)\notin \{0,1\}$ flip a single biased coin in the first phase. 
To derandomize these coin flips we iterate through the $C$ color classes. In iteration $i$ all nodes with color $i$ determine the outcome of their coin flip (or equivalently the value of $X_v$) such that the total expected size at the end of the second phase (over the randomness of coin flips of nodes with colors $i+1,\ldots,C$) does not increase by more than $1/n^{7}$. Fix some iteration $i$, assume that all nodes with  color $<i$ have already fixed their coin flips and let $\theta$ denote the fixed coin flips of all nodes with color $<i$.  We describe how a single node $v$ with color $i$ determines $X_v$.  Let each node $u\in N(v)$ compute its expected fractional value after the second phase, once conditioned on all fixed coin flips and $X_v=x(v)/p(v)$ and once conditioned on all fixed coin flips and $X_v=0$. 
	\begin{align}
          \alpha_{u,1}&=\mathbb{E}[Z_u  \mid \theta_u\wedge X_v=x(v)/p(v) ] & \alpha_{u,0} &=\mathbb{E}[Z_u \mid \theta_u\wedge X_v=0 ] 
	\end{align}
	        Let $\tilde{\alpha}_{u,1}$ and $\tilde{\alpha}_{u,0}$ denote $\alpha_{u,1}$ and $\alpha_{u,0}$ rounded up to the next multiple of $1/n^{10}$, respectively. Each node can compute both values in a single round by learning its neighbors already fixed coin flips (fixed coin flips of nodes that are further away do not influence the expectation of $Z_u$), $x$-values and probabilities; all these values can be learned in a single round as all values are transmittable. Then, each $u\in N(v)$ sends $\tilde{\alpha}_{u,j}$, $j=0,1$ to $v$ and $v$ computes the values $\widetilde{A}_{v,1}$ and $\widetilde{A}_{v,0}$:
        \begin{align}
          &\widetilde{A}_{v,1} =\sum_{u\in N(v)}\tilde{\alpha}_{u,1} \quad\quad &&  A_{v,1} =\sum_{u\in N(v)}\alpha_{u,1}\\
          &\widetilde{A}_{v,0} =\sum_{u\in N(v)}\tilde{\alpha}_{u,0} \quad\quad &&  A_{v,0} =\sum_{u\in N(v)}\alpha_{u,0}
	\end{align}
	 The values values $A_{v,1}$ and $A_{v,0}$ are only defined for analysis purposes and cannot be computed by node $v$.
Then node $v$ fixes its biased coin (or equivalently sets $X_v$) as follows: 
\begin{align}  \label{eqn:decision}
X_v=
\left\{\begin{array}{lr}
\frac{x(v)}{p(v)}, & A_{v,1}<A_{v,0} \\
0, & \text{otherwise}.
\end{array}\right.
\end{align} 
The decision of $v$ only depends on values and probabilities in $v$'s $2$-neighborhood and none of these values are changed in iteration $i$ as due to the distance two coloring property no other node in the $2$-neighborhood of $v$ decides on its value in iteration $i$. Thus, the expected size after iteration $i$ is the same as if nodes with color $i$ decided with the same protocol sequentially. 
It remains to show that the expected size does not increase by more than $1/n^7$ by fixing the biased coin of a single node---without the rounding $\alpha_{u,j}$ to $\tilde{\alpha}_{u,j}$ there would not be an increase at all. Let the random variable $Z$ be the size of the CFDS after the second phase, for $u\in V$ let the random variable $Z_u$ be the value of node $u$ after the second phase and  let $b_v$ be the value that $v$ chose for $X_v$ in (\ref{eqn:decision}).  We show that the expected size of the dominating set increases at most by  $1/n^8$ by the choice $X_v=b_v$.
\begin{align*}
  \mathbb{E}[Z \mid \theta \land X_v = b_v] - \mathbb{E}[Z \mid \theta]
  &= \sum_{u \in V} \left( \mathbb{E}[Z_u \mid \theta \land X_v = b_v] - \mathbb{E}[Z_u \mid \theta] \right) \\
  & \overset{(*)}{=} \sum_{u \in N(v)}\left( \mathbb{E}[Z_u \mid \theta \land X_v = b_v] - \mathbb{E}[Z_u \mid \theta]\right)\\
  & = \mathbb{E}\left[\sum_{u \in N(v)} Z_u \mid \theta \land X_v = b_v\right] - \mathbb{E}\left[\sum_{u \in N(v)} Z_u \mid \theta\right]\\
  & = A_{v, b_v} - \mathbb{E}\left[\sum_{u \in N(v)} Z_u \mid \theta\right]\\
  &\leq \frac{1}{n^9} + \widetilde{A}_{v, b_v} - \mathbb{E}\left[\sum_{u \in N(v)} Z_u \mid \theta\right]\\
  &\stackrel{(**)}{=} \frac{1}{n^9} + \min\left\{\widetilde{A}_{v, 0}, \widetilde{A}_{v, 1}\right\} - \mathbb{E}\left[\sum_{u \in N(v)} Z_u \mid \theta\right]\\
 &\leq  \frac{2}{n^9} + \min\left\{A_{v, 0}, A_{v, 1}\right\} - \mathbb{E}\left[\sum_{u \in N(v)} Z_u \mid \theta\right]\\
 &\stackrel{(***)}{\leq}  \frac{2}{n^9} \leq \frac{1}{n^8}~,
\end{align*}
where $(*)$ follows as $X_v$ only influences the direct neighbors of $v$ (for all other nodes we have $\mathbb{E}[Z_u \mid \theta \land X_v = b_v] - \mathbb{E}[Z_u \mid \theta] = 0$), 
$(**)$ follows due to the choice of $b_v$
and $(***)$ holds by the law of total expectation.
Thus the total increase of the dominating set until all nodes have determined their value is at most $n\cdot\frac{1}{n^8}\leq \frac{1}{n^7}$. 
The fractionality of the resulting dominating set follows from the fractionality of the abstract randomized rounding algorithm.	
\end{proof}

We want to use \Cref{lem:derandomizationDelta} and the one shot or the two factor rounding process from \Cref{sec:derandomizationN} to increase the fractionality of a dominating set. However, if we applied \Cref{lem:derandomizationDelta} to $G$ itself we would inherently have a runtime of $O(\Delta^2)$ rounds as we need $O(\Delta^2)$ colors in a distance two coloring of $G$. Instead we apply \Cref{lem:derandomizationDelta} to a graph that we obtain by modifying the so called bipartite representation of $G$. 
\paragraph{Bipartite Representation:}
The \emph{bipartite representation}  $B_G$ of a graph $G=(V,E)$ and a CFDS $(x,c)$ is given through a slight adaption of the \emph{bipartite double cover} of $G$ together with a CFDS $(\tilde{x},\tilde{c})$. The vertex set of $B_G$ consists of two copies $V_L$ and $V_R$ of $V$. Two nodes $u\in V_L, u'\in V_R$ are connected in $B_G$ if and only if their counterparts in $V$ are connected in $G$ or if they stem from the same node.  The CFDS $(\tilde{x},\tilde{c})$ of $B_G$ is given through
\begin{align*}
\tilde{x}(v) = \left\{\begin{array}{lr}
        x(v) & v\in V_R\\
        0, & v\in V_L 
        \end{array}\right. 
				\hspace{1cm}
\tilde{c}(v) = \left\{\begin{array}{lr}
        0 & v\in V_R\\
        c(v), & v\in V_L. 
        \end{array}\right. 
\end{align*}
This natural identification can be seen as splitting each node $v\in V$ into two nodes, one node in $V_L$ that takes care of the constraint of $v$ and one node in $V_R$ that takes the value of $v$. 
The input to the algorithms in the upcoming \Cref{lem:deltaDetOneRounding,lem:deltaDetTwoRounding} is  a graph $G$ together with a constraint fractional dominating set $(x',c')$. First, the algorithms compute the bipartite representation $B_G$, $(\tilde{x}',\tilde{c}')$. Then they modify  $B_G$ and  $(\tilde{x}',\tilde{c}')$ to obtain a new  bipartite graph $\mathcal{B}$ with a new CFDS $(x,c)$---the modification consists of the removal of edges, splitting of nodes, changing constraints and so on. The (virtual) graph $\mathcal{B}$ admits a distance two coloring with few colors and at the same time ensures small $\Pr(E_v=1)$ for all $v\in V(\mathcal{B})$ which is important to obtain a small dominating set in the end. 
Then we apply \Cref{lem:derandomizationDelta} to this modified graph $\mathcal{B}$ and a suitable chosen instantiation of the abstract rounding process from \Cref{sec:abstractRounding}---the instantiations are almost identical to the one shot rounding and factor two rounding from \Cref{sec:derandomizationN}. The graph $\mathcal{B}$ and the instantiation of the rounding process are such that we can compute the necessary distance two coloring of nodes in $\mathcal{B}$ with $O(\Delta\polylog \Delta)$ colors and at the same time we ensure that $\Pr(E_v=1)$ is small for all nodes of $\mathcal{B}$. 
Afterwards use the resulting CFDS on $\mathcal{B}$ to construct an FDS on $G$ with the properties of the respective lemma. 
To make the derandomization work we need the following two resutls.

\begin{theorem}{Chernoff Bound \cite{concentration09}}\label{thm:chernoffBound}
 Let $\lambda>0$ and $X=\sum_{i=1}^nX_i$ where $X_i, 1\leq i\leq n$ are independently distributed in $[0,\lambda]$. Then for $\delta>0$
we have $\Pr(X<(1-\delta)\E[X])\leq e^{-\frac{\delta^2}{2\lambda}\E[X]}$~.
\end{theorem}

\begin{lemma}[Bipartite Coloring]
\label{lem:bipatiteColoring}
   There is a deterministic \CONGEST algorithm that, given a bipartite graph $B=(V_L\cup V_R, E)$  computes a distance two coloring of $V_R$ with $O(\Delta_L \cdot \Delta_R)$ colors in $O(\Delta_L \cdot \Delta_R + \Delta_L \cdot \log^*n)$ round where $\Delta_L$ and $\Delta_R$ are the maximum degree of nodes in $V_L$ and $V_R$, respectively. 
\end{lemma}
\begin{proof}
Let $G=(V_R, F)$ be the graph on $V_R$ that is obtained by connecting any two nodes in $V_R$ that have a common neighbor in $B$.
The coloring of $B$ is obtained by computing an $O(\Delta_L\cdot \Delta_R)$-coloring of $G$: First compute a $O((\Delta_L\cdot \Delta_R)^2)$-coloring of $G$ with Linial's algorithm. One round of this algorithm on $G$ can be simulated in $O(\Delta_L)$-rounds in $B$ yielding a runtime of $O(\Delta_L\cdot \log^*n)$. Afterwards run the \CONGEST algorithm from \cite{BEG18} on $G$ yielding a $O(\Delta_L\cdot \Delta_R)$-coloring in $O(\Delta_L\cdot \Delta_R)$ rounds. The algorithm of \cite{BEG18} can be reformulated such that every node $v\in V_R$ uses the  $O((\Delta_L\cdot \Delta_R)^2)$-coloring to compute (without communication) a sequence  $p_v(0),p_v(1),\ldots,p_v(q-1)$ of colors of length $q=O(\Delta_L\cdot \Delta_R)$. If $v$ is still uncolored at the beginning of round $i$ it tries to obtain color $c=p_v(i)$. To do so it tests whether any neighbor in $G$ tries to get color $c$ in round $i$ or is already colored with $c$, if not it keeps color $c$, otherwise it discards the color and tries again in the next round. The sequences are chosen such that every node is colored after $q$ rounds. The respective 'test' can be performed in the communication network $B$ in $O(1)$ rounds by letting nodes in $V_L$ performing the test for their neighbors in $V_R$. 
\end{proof}

\begin{restatable}[Deterministic One Shot Rounding]{lemma}{lemDeltaDetOneRounding}
  \label{lem:deltaDetOneRounding}
Let $F>0$ be an integer. There is a determ. algorithm that,
given a graph $G$ with a $1/F$-fractional dominating set of size $A$,
 computes an integral dominating set with size at most $\ln \Deltab\cdot A+n/\Deltab+n^{-6}$ and has round complexity $O(F\cdot \Delta+F\cdot \logstar n)$ in the \CONGEST model and round complexity $O(F\cdot \Delta+\logstar n)$ in the \LOCAL model.
\end{restatable}

\begin{proof}
To apply \Cref{lem:derandomizationDelta} efficiently we first construct a new (virtual) graph $\mathcal{B}$ that admits a distance two coloring with few colors and at the same time ensures small $\Pr(E_v=1)$ for all $v\in V(\mathcal{B})$. 

  \mypara{Constructing Graph $\mathcal{B}$:} 
Let $x'(v)$ denote the values of the $1/F$-fractional dominating set. First, set $x(v)=\min\{1,\ln \Deltab \cdot x'(v)\}$ and set $c(v)=1$ for all nodes.  Note that we actually increase each $x(v)$ such that it is transmittable which increases each value by at most $n^{-10}$ and therefore, summing over all nodes, will give an overall increase to the resulting FDS of at most $n^{-9}$. 
Then we build the Bipartite Representation $B_G, (x,c)$ (reuse of notation for the dominating set) of $G$ with $(x,c)$. For every node the fractional value is either $0$ or at least $1/F$. Therefore for each node $v$ on the left hand side, there is a set of at most $F$ nodes that cover $v$. For each node we find such a set. We then remove all edges that connect the node to nodes outside the set. With this we reduce the degree on the left hand side to $F$.

		\mypara{Instance of the Abstract Rounding Algorithms:} 
Let  $p(v)=x(v)$ for all $v\in V(\mathcal{B})$ and consider the abstract randomized rounding process with these values on $\mathcal{B}$. 

\mypara{Bounding \boldmath$\Pr(E_v=1)$:}
For $v\in V(\mathcal{B}$ let $E_v$ be the random variable that is one if $v$'s constraint is violated after the first phase of the abstract randomized rounding process and zero otherwise. If $v$ is on the right hand side of $\mathcal{B}$ its constraint is zero and $\Pr(E_v=1)=0$. So let $v$ be a node with $c(v)=1$. If $v$ has a neighbor $u$ with $x(u)=1$ then $v$'s constraint is satisfied. Otherwise we have $x(u)\geq \ln\Deltab \cdot x'(u)$ for all neighbors of $v$ in $\mathcal{B}$. Thus we can bound the probability that $v$ is uncovered after the first step as follows. 
\begin{align*}
  \mathbb{P}(E_v) &= \prod_{u \in N(v)} (1-x_{u} \ln \Deltab) 
                                  \leq \prod_{u \in N(v)} \mathrm{e}^{-x_{u} \ln \Deltab} \\
                                  & = \mathrm{e}^{\sum_{u \in N(v)} -x_{u} \ln \Deltab} 
                                  = \mathrm{e}^{-\ln (\Deltab) \sum_{u \in N(v)} x_{u} } 
                                  \overset{(*)}{\leq} \mathrm{e}^{-\ln \Deltab} 
                                  = \frac{1}{\Deltab}~.
 \end{align*}
 The inequality $(*)$ holds since the sum over all $x'(u)$ in the closed neighborhood of $v$ is at least 1 because we used a valid fractional dominating set as an input.

\mypara{Applying the Derandomization Lemma and Building FDS on \boldmath$G$:}
Using \Cref{lem:bipatiteColoring} we can color the right hand side nodes $V_R$ with $O(F\cdot \Delta)$ colors in $O(F\cdot \Delta+F\cdot\logstar n)$ rounds. We can then use this coloring together with \Cref{lem:derandomizationDelta} to generate a dominating set of $\mathcal{B}$ of size
\[\ln \Deltab \cdot A + \sum_{v \in V} \Pr(E_v = 1) + \frac{1}{n^8}+\frac{1}{n^9}\leq \ln \Deltab \cdot A + \frac{n}{\Deltab} + \frac{1}{n^7}~.\]
Here the $n^{-9}$ term is due to the rounding of the values to transmittable values that we performed at the beginning.
The FDS on $\mathcal{B}$ induces an FDS on $G$ by reverting the bipartite representation where a node sets its value to the maximum of the values of its two copies. 
 The fractionality of the resulting dominating set follows from the fractionality of the abstract randomized rounding algorithm.
The result in the \LOCAL model follows as the required distance two coloring of $V_R$ can be computed in $O(F\cdot \Delta+\logstar n)$ rounds, e.g., with the algorithm of \cite{BEK15}. 
\end{proof}

\begin{restatable}[Deterministic Factor Two Rounding]{lemma}{lemDeltaDetTwoRounding}
  \label{lem:deltaDetTwoRounding}
Let $\eps>0$, $r\geq 256\eps^{-3}\log \Deltab$. There is a deterministic \CONGEST algorithm that, given
a graph $G$ with a $1/r$-fractional dominating set of size $A$,
computes a $\min\{2/r,1/F\}$-fractional dominating set with size $(1+\eps)A+n/\Deltab^4+n^{-6}$ and has round complexity $O(\eps^{-2}\Delta\cdot \log \Delta +\eps^{-2}\log \Delta\cdot \logstar n)$ in the \CONGEST model and round complexity $O(\Delta\cdot \polylog \Delta +\logstar n)$ in the \LOCAL model.
\end{restatable}

\begin{proof} The fractionality of the input dominating set is increased if all nodes with value $<2/r$ double their value with probability $1/2$ and set it to $0$ otherwise. 
This can be achieved through an instance of the abstract randomized rounding process and one could use \Cref{lem:derandomizationDelta} to derandomize this process. However, \Cref{lem:derandomizationDelta} requires that all nodes that actually take part in the rounding process are distance two colored. Such a coloring of $G$ would have $O(\Delta^2)$ colors and the runtime of the derandomization lemma crucially depends on the number of colors. Thus, we first move to the bipartite representation of $G$, remove edges and split nodes of this representation and obtain a graph $\mathcal{B}$. Then we derandomize the abstract rounding process on $\mathcal{B}$. The properties of $\mathcal{B}$ ensure 
1) that we only need $\Delta\polylog\Delta$ colors in a distance two coloring of the nodes yielding a $\Delta\polylog \Delta$ runtime and 2) that the probability of a node being uncovered in the randomized rounding process is small; the latter is important as the approximation factor of the FDS at the end of derandomization depends on these probabilities.

  \mypara{Constructing Graph $\mathcal{B}$:}
  Let $x'(v)$ denote the value of the input $1/r$-fractional dominating set. Let $x(v)=\min\{1,(1+\eps)x'(v)\}$. We increase each value of $x$ such that it is transmittable, this increases each value by at most $\frac{1}{n^{10}}$ and therefore gives an overall increase to the FDS of at most $\frac{1}{n^9}$.
  We define the constraints of the graph by $c(v) = 1$, and let $s=64\eps^{-2} \ln \Deltab$. Next, we modify the bipartite Representation $B_G$ of $G$ to obtain a bipartite graph $\mathcal{B} = (U_L \cup U_R, E_\mathcal{B})$. Begin with the bipartite representation $B_G=(V_L\cup V_R, E)$ of $G$. Each node $v \in V_L$ is represented by $k$ nodes $v_1, \dots, v_k$ in $U_L$. The edges of $v$ are distributed among these nodes as follows: $v_1$ is connected to all nodes in $U_R$ that correspond to nodes $u \in V_R$ with $\{v,u\}\in E$ and $x(u)\geq 2/r$~. 
  If the number of remaining edges of $u$ is less than $s$ all of them are also given to $v_1$ (in this case we have $k=1$). Otherwise the remaining edges of $v$ are split between $v_2,\ldots,v_k$ such that each node gets between $s$ and $2s$ edges. Note that all edges in $\mathcal{B}$ are edges of the original graph and any \CONGEST algorithm on $\mathcal{B}$ can be executed without any overhead in the original network graph $G$. For $v\in U_R$ we set $c(v)=0$ and for $v\in U_L$ we set the constraints such that they are satisfied by the $x'$ values of their neighbors meaning $c(v)=\max\big\{1,\sum_{u\in N_{\mathcal{B}}(v)}x'(u_G)\big\}$ where $u_G$ is the node that corresponds to $u$ in $G$~. For the fractional values we define $x(v) = 0$ for $v\in U_L$ and $x(v) = x(v_G)$ for $v \in U_R$.
  
		\mypara{Instance of the Abstract Rounding Algorithms:}
	We define the following instance of the abstract randomized rounding algorithm on $\mathcal{B}$. For $v\in U_L$ we set $p(v)=1$. For $v\in U_R$ we set
\begin{align}
  p(v) = \left\{\begin{array}{lr}
			        \frac{1}{2}, & \text{if $x(v)<2/r$,}\\
        1, & \text{if $x(v)\geq 2/r$}~. 
  \end{array}\right.
\end{align}

\mypara{Bounding \boldmath$\Pr(E_v=1)$:}
For $v\in U_L\cup U_R$ let $E_v$ be the random variable that is one if $v$'s constraint is violated after the first phase of the abstract randomized rounding process and zero otherwise. If $v\in U_R$ we clearly have $\Pr(E_v=1)=0$ as $c(v)=0$. For a $v \in U_L$ we have the following cases: 

In the first case $v$ is a $v_1$ type node, meaning all neighbors of $v_B$ are connected to $v$. 
If no $u\in N(v)$ has $x(u)<2/r$ none of the neighbors takes part in the rounding and we have $E_v=0$ as none of the neighbors change their value and $v$ was covered before the rounding process. 
Now, assume that $v_1$ has neighbors $u\in N(v)$ with $x(u)<2/r$.  Let $S=\{u\in N(v) \mid x(u)<2/r\}$ be the neighbors of $v$ that participate in the rounding process and $T=N(v)\setminus S$. As $v$ is a $v_1$ type we have $|S|\leq s$ and thus the contribution of participating nodes is at most $\sum_{u\in S}x'(u)\leq \sum_{u\in S}x(u) \leq s\cdot 2/r$.
 As $v$ is covered before the rounding we obtain that $\sum_{u\in T}x'(u)\geq 1-2s/r$. 
If $x(u)=1$ for some $u\in T$ we clearly have $E_v=0$, otherwise we have $x(u)\geq (1+\eps)x'(u)$ and we obtain  $\Pr(E_v=1)=0$ as 
\begin{align}
\sum_{v\in T}x(u)\geq \sum_{v\in T}(1+\eps)x'(u)\geq (1+\eps)(1-2s/r)\stackrel{(*)}{\geq} (1+\eps)\left(1-\frac{\eps}{1+\eps}\right) = 1~, 
\end{align}
where $(*)$ follows because $2s/r\leq \frac{2\cdot 64 \varepsilon^{-2} \ln \Delta}{256 \varepsilon^{-3} \ln \Delta} = \frac{\varepsilon}{2} \leq \frac{\eps}{1+\eps}$~. 

In the second case $v$ is a $v_j, j\geq 2$ type node meaning all neighbors of $v$ have a fractional value of less than $2/r$. Recall that for $u\in N(v)$ the random variable $X_u\in [0,4/r]$ denotes the value of $u$ after the first phase of the rounding process. Define $X=\sum_{u\in N(v)}X_u$ and note that due to linearity of expectation $\E[X]=c(v)(1+\eps)$. 
With $\delta=1-\frac{1}{1+\eps}=\frac{\eps}{1+\eps}$, $\eps<1$, $c(v) \geq s/r$, $\lambda=4/r$ and the Chernoff Bound from \Cref{thm:chernoffBound} we obtain that 
\begin{align*}
  \Pr(E_v =1 ) &=\Pr(X<c(v))=\Pr(X<(1-\delta)\E[X])\leq e^{-\frac{\delta^2}{2 \lambda} \E[X]} \\
               &= e^{-\frac{\delta^2}{8}\cdot r\cdot c(v)(1+\eps)} \leq e^{-\frac{\eps^2}{(1+\eps)8}\cdot r\cdot c(v)}=e^{-\frac{\eps^2}{16}\cdot r\cdot c(v)}\leq e^{-\frac{\eps^2}{16}\cdot s}= \Deltab^{-4}~.
\end{align*}
\mypara{Applying the Derandomization Lemma and Building FDS on \boldmath$G$:}
Note that for all $u\in V_L$ we have $p(u)=1$ and each $u\in V_L$ is connected to at most $O(s)$ nodes $u'$ in $V_R$ with $p(u')\notin \{0,1\}$. Thus we can use \Cref{lem:bipatiteColoring} to compute a distance two coloring of $S=\{v\in V_R \mid p(v)\notin \{0,1\}\}$ with  $O(s \cdot \Delta)=O(\eps^{-2}\log \Delta \cdot \Delta)$ colors in $O(\eps^{-2}\cdot \log\Delta\cdot \Delta + \eps^{-2}\log\Delta \cdot \log^*n)$. 

 We now use  \Cref{lem:derandomizationDelta} with graph $\mathcal{B}$, the constraint fractional dominating set $(x,c)$ with transmittable values and the distance two coloring of $S$ to obtain a FDS of $\mathcal{B}$ with size at most
\[\ln \Deltab \cdot A + \sum_{v \in V} \Pr(E_v = 1) + \frac{1}{n^8}+\frac{1}{n^9}\leq \ln \Deltab \cdot A + \frac{n}{\Deltab^4} + \frac{1}{n^7}~.\]
Here the $n^{-9}$ term is due to the rounding of the values to transmittable values that we performed at the beginning.
The FDS on $\mathcal{B}$ induces an FDS on $G$ by reverting the bipartite representation where a node sets its value to the maximum of the values of its two copies. 
 The fractionality of the resulting dominating set follows from the fractionality of the abstract randomized rounding algorithm.
The result in the \LOCAL model follows as the required distance two coloring of $V_R$ can be computed in $O(s\cdot \Delta+\logstar n)$ rounds, e.g., with the algorithm of \cite{BEK15}. 
\end{proof}


\subsection{Proof of \Cref{thm:mainThmN}}
\label{sec:proofs}
We use the deterministic rounding results (\Cref{lem:oneShotRounding,lem:factorTwoRounding,lem:deltaDetOneRounding,lem:deltaDetTwoRounding}) to prove our main results.

\begin{proof}[Proofs of \Cref{thm:mainThmN}/\Cref{thm:mainThmDelta}/\Cref{corr:mainCorDelta}]
  Set $\rho=\log (\Delta/\eps)$, $\eps_1=\min\{\eps/16,1/4\}$, $\eps_2=\eps_1/(100\rho)$ and  $F= 256 \varepsilon^{-3} \log \Deltab$. Then the algorithm for all three results consist of three parts.
  \begin{itemize}
    \item \textbf{Part I:} Use \Cref{lem:fractionalApproximation} to build an initial FDS with fractionality $\eps_1/\poly \Delta$ and $(1+\eps_1)$-approximation guarantee,
    \item \textbf{Part II:} iterate (with $\eps_2$) \Cref{lem:factorTwoRounding} or \Cref{lem:deltaDetTwoRounding} $\rho$ times to increase the fractionality by a factor two until the fractionality is $F$ and
    \item \textbf{Part III:} use \Cref{lem:oneShotRounding} or \Cref{lem:deltaDetOneRounding} to round the FDS to an integral one.
  \end{itemize}
  In this \Cref{lem:factorTwoRounding} and \Cref{lem:oneShotRounding} are used to prove \Cref{thm:mainThmN}. \Cref{lem:deltaDetTwoRounding} and \Cref{lem:deltaDetOneRounding} are used to prove \Cref{thm:mainThmDelta}. 
  We begin with a $(1+\eps_1)$ approximation guarantee, roughly loose a $(1+\eps_1/\rho)$ factor in the approximation guarantee in each iteration of second part and roughly a $\ln\Deltab$ factor in the third part, which heuristically bounds the approximation guarantee as follows
	\[(1+\eps_1)(1+\eps_2/10)(1+\eps_1/(100\rho))^\rho \ln\Deltab\leq (1+\eps)\ln \Deltab~.\]


To study the runtime of the algorithm   let $T_{col}(\Delta_L,\Delta_r)=O(\Delta_L\cdot\Delta_r+\Delta_L\cdot\logstar n)$ be the time to distance two color the right hand side of a bipartite graph with maximum degree $\Delta_L$ on the left hand side and $\Delta_R$ on the right hand side. Further note that $F = O(\varepsilon^{-3} \log \Delta)$ and let $s = \varepsilon^{-2} \log \Delta$. 
We obtain the following runtimes of the different parts of the algorithm

\begin{itemize}
  \item Part I: \Cref{lem:fractionalApproximation} $O(\varepsilon^{-4} \log^2\Delta)$
  \item Part II: $\rho$ iterations of \Cref{lem:deltaDetTwoRounding} take  $O(\rho\cdot s\cdot  T_{col}(s,\Deltab)$
  \item Part III: \Cref{lem:deltaDetOneRounding}: $O(F\cdot \Delta +  T_{col}(s,\Deltab)$~.
\end{itemize}

For \Cref{thm:mainThmN} we obtain a runtime of 
\begin{align*}
O(\eps^{-4}\log^2\Delta)+(\rho+1)\cdot 2^{O(\sqrt{\log n \log\log n}} & =O(\eps^{-4}\log^2\Delta)+2^{O(\sqrt{\log n \log\log n}}~.
\end{align*}
For \Cref{thm:mainThmDelta} we obtain a runtime of 
\begin{align*}
  &O(\varepsilon^{-4} \log^2 \Delta + (\Delta \varepsilon^{-2} \log \Delta + \varepsilon^{-2} \log \Delta \cdot \log^*n)\cdot \log \Delta + \Delta \varepsilon^{-3} \log \Delta + \varepsilon^{-3} \log^*n)\\
  =&O(\varepsilon^{-4} \log^2 \Delta + \Delta \varepsilon^{-2} \log^2 \Delta + \varepsilon^{-2} \log^2 \Delta \cdot \log^*n + \Delta \varepsilon^{-3} \log \Delta + \varepsilon^{-3} \log^*n)
\end{align*}

The exact upper bound on the size of the dominating set after part II is given by 
\begin{align*}
  (1 + \varepsilon_1) (1 + \frac{\varepsilon_1}{\rho})^\rho \cdot A + \sum_{i=0}^{\rho -1} \left(1+ \frac{\varepsilon}{\rho}\right)^i \left(\frac{n}{\Delta^4} + \frac{1}{n^6}\right)
  &\leq(1 + \varepsilon_1) e^{\varepsilon_1} \cdot A + \sum_{i=0}^{\rho -1} \left(1+ \frac{\varepsilon}{\rho}\right)^i \left(\frac{n}{\Delta^4} + \frac{1}{n^6}\right)\\
  &\leq(1 + \varepsilon) \cdot A + \sum_{i=0}^{\rho -1} \left(1+ \frac{\varepsilon}{\rho}\right)^i \left(\frac{n}{\Delta^4} + \frac{1}{n^6}\right)\\
  &=(1 + \varepsilon) \cdot A + \frac{\left(1+ \frac{\varepsilon_1}{\rho}\right)^\rho -1}{\left( 1 + \frac{\varepsilon_1}{\rho} \right) -1} \left(\frac{n}{\Delta^4} + \frac{1}{n^6}\right)\\
  &\leq(1 + \varepsilon) \cdot A + 2 \rho \left(\frac{n}{\Delta^4} + \frac{1}{n^6}\right)
\end{align*}
The exact upper bound on the size of the dominating set after part III is given by
\begin{align*}
  &\left((1 + \varepsilon) \cdot A + 2 \rho \left(\frac{n}{\Delta^4} + \frac{1}{n^6}\right)\right) \ln \Deltab + \frac{n}{\Deltab} + \frac{1}{n^5}
\end{align*}
Note that the approximation guarantee is at most $(1 + \varepsilon) (2+\ln\Deltab)$ as for small constant $\Delta$ part II is not executed at all and for $\eps>1/\polylog \Delta$ the term can be bounded as claimed. 
\end{proof}


\section{Connected Dominating Set}
\label{app:CDS}
\label{sec:CDS}
In the following, we discuss how to extend the dominating set
algorithm \Cref{sec:derandomizationN} to obtain a logarithmic
approximation of the connected dominating set problem. We first recall
how to extend a given dominating set $S$ of a connected graph $G$ to a
connected dominating set (CDS) $\bar{S}$ of $G$. A proof for the
following statement for example appears in \cite{dubhashi05} (see also
\cite{guhakhuller98}).

\begin{claim}\label{claim:basicCDS}\cite{dubhashi05,guhakhuller98}
  Let $G=(V,E)$ and let $S\subseteq V$ be a dominating set of $G$.
  Further, let $G_S=(S,E_S)$ be a graph defined on the nodes in $S$,
  where there is an edges $\set{u,v}\in E_S$ for $u,v\in S$ if
  $d_G(u,v)\leq 3$. Then, $G_S$ is connected if and only if $G$ is
  connected.
\end{claim}

\paragraph{Algorithm Outline:} Due to \Cref{claim:basicCDS} we can  compute a CDS
$\bar{S}$ of size at most $3|S|$ by first choosing a spanning tree
$T_S$ of $G_S$ and by then replacing each edge of $T_S$ by a path of
length at most $3$ in $G$ and adding the at most $2$ inner nodes of
the path to $\bar{S}$. The time to compute a spanning tree in a
distributed way is linear in the diameter of the graph. However, we
can replace the spanning tree $T_S$ by a slightly larger sparse
spanning subgraph that can be computed efficiently in the distributed
setting. In the \LOCAL model, this is straightforward and has for
example be done by Dubhashi et al.\ in \cite{dubhashi05}. In the
\CONGEST model, Ghaffari \cite{ghaffari14_CDS} showed that sparse
skeleton construction of Pettie \cite{pettie10} can be adapted to work
on the graph $G_S$ with small messages. Note however that the
algorithm of \cite{pettie10} and the algorithm of \cite{ghaffari14_CDS} are randomized. Alternatively, one could
try to use the more basic spanner algorithm of Baswana and Sen
\cite{baswana07}, this would however come at a cost of a factor
$O(\log n)$ in the number of edges of the spanning subgraph and thus
in the number of nodes of the constructed CDS. In the derandomized
version of the spanner algorithm of \cite{baswana07}, which has been
presented by Ghaffari and Kuhn in \cite{DISC18_DomSet}, this
additional cost even increases to an $O(\log^2 n)$ factor. In order to
avoid this, we will first show how to apply a clustering of the graph
$G_S$ to reduce the problem of finding a spanning subgraph on $G_S$ to
finding a spanning subgraph on a smaller graph $G_S'$ of size at most
$O(|S|/\log^2 n)$. We then show that it is possible to use the
technique of \cite{ghaffari14_CDS} in order to efficiently run the
deterministic spanner algorithm of \cite{DISC18_DomSet} on $G_S'$ in the \CONGEST model,
which then gives a CDS $\bar{S}$ of size $|\bar{S}|=O(|S|)$.

\paragraph{Reducing the Problem Size:}
We first describe how to reduce the problem to finding a sparse subgraph
on a smaller graph $G_S'$. We first compute a subset $S'$ of $S$ of
size at most $S/(c\log^2 n)$ for a sufficiently large constant $c$. By
using the \CONGEST model ruling set algorithm of
\cite{awerbuch89,CONGEST_rulingsets}, one can deterministically
compute a subset $S'$ in time $O(\log^3 n)$ such that any two nodes
$u,v\in S'$ are at distance $d_G(u,v) \geq c'\log^2 n$ for a chosen constant
$c'>0$ and such that every node in $S\setminus S'$ has some node in
$S'$ within distance $O(\log^3 n)$ in $G$. We then partition the set $S$ into $|S'|$ clusters such that each node $v\in S'$ is the center of a cluster and such that each node $u\in S\setminus S'$ joins the cluster of the closest cluster center in the graph $G_S$. In addition, for each cluster $C\subseteq S'$, we construct a subtree of $G$ that spans all the nodes in $S$ and that contains at most $3|C|$ nodes of $G$.

The clusters are constructed in a BFS-type fashion as follows. The
algorithm runs in phases $i=1,2,\dots$. For a node $v\in S'$, let
$C_v$ be $v$'s cluster and let $T_v$ be the subtree of $G$ that spans
$C_v$. At the beginning each cluster $C_v$ (and thus also each tree
$T_v$) only consists of the node $v$ itself. Each phase then consists
of three rounds, In the first $2$ rounds of the phase, we add nodes
from $V\setminus S$ (i.e., nodes that are not in the given dominating
set) to the clusters. Concretely, in the first round of phase $i$,
each node $w\in V\setminus S$ that is not yet in a cluster tree and
that neighbors a node $u\in S$ that has been added to a cluster $C_v$
in phase $i-1$ adds the edge $\set{w,u}$ to the tree $T_v$ and joins
the tree $T_v$. If there are several choices for the node $u$, node
$w$ chooses one of them arbitrarily. In the second round of phase $i$,
every node $w'\in V\setminus S$ that is not yet in a cluster tree and
that has a neighbor $w\in V\setminus S$ that joined a cluster tree
$T_v$ in the first round of phase $i$ joins the cluster tree $T_v$ of
$w$ and adds the edge $\set{w,w'}$ to $T_v$. Again, if $w'$ can choose
among different neighbors $w$, it chooses one arbitrarily. Finally, in
the third round of the phase, every node $u\in S$ that is not yet in a
cluster tries to join a cluster. A node $u\in S$ can join a cluster in
this step if it either has a neighbor $x\in S$ that joined a cluster
in the previous phase or if it has a neighbor $w\in V\setminus S$ that
joined a cluster in one of the first two rounds of phase $i$. Again,
if $u$ can choose among different nodes, to join a cluster, it picks
one arbitrarily and adds an edge to it to the respective cluster
tree. The clustering algorithm ends as soon as all nodes are in some
cluster. At the end, in every cluster tree $T_v$, we remove all nodes
$w\in V\setminus S$ for which there is no node of $S$ in its subtree
(i.e., we prune the tree to only contain the nodes that are necessary
to connect the actual cluster nodes). Finally, given the clustering,
we define the cluster graph $G_S'$ as follows. The nodes of $G_S'$ are
the $|S'|$ clusters of the clustering and two clusters are neighbors
in $G_S'$ if they contain nodes $u\in S$ and $v\in S$ that are
neighbors in $G_S$. The following lemma analyzes the
relevant properties of the described clustering algorithm.

\begin{lemma}\label{lemma:CDSclustering}
  The above clustering algorithm runs in time $O(\log^3 n)$ in the
  \CONGEST model and it constructs cluster trees of radius
  $O(\log^3 n)$. The total number of clusters is at most
  $|S|/(c\log^2 n)$ for a constant $c>0$ that can be chosen and the
  number of nodes in all the cluster trees is less than
  $3|S'|$. Further, the cluster graph $G_S'$ is connected if and only
  if $G$ is connected.
\end{lemma}
\begin{proof}
  Consider some node $u\in S$ and let $v\in S'$ be the closest node to
  $u$ in graph $G_S$ and assume that the distance between $u$ and $v$
  is $d$. It follows directly from the construction of the clustering
  (and by induction on the number phases) that $u$ joins some cluster
  in phase $d$ of the clustering algorithm. We next show that the
  number of phases and thus the maximum cluster tree diameter is at
  most $O(\log^3 n)$. The ruling set constructiong of $S'$ guarantees
  that every node $u\in S\setminus S'$ has some node $v$ in $S'$ at
  distance at most $O(\log^3n)$ in $G$. We need to show that the
  distance between $u$ and $v$ is also at most $O(\log^3 n)$ in
  $G_S$. To see this, consider a shortest path $P$ in $G$ connecting
  $u$ and $v$. Assume that the length of $P$ is $d=O(\log^3n)$. Each
  of the nodes of $P$ is dominated by some node in $S$. Assume that
  two neighbors $x$ and $y$ on $P$ are dominated by nodes $s_u\in S$
  and $s_v\in S$. Because the distance between $s_u$ and $s_v$ in $G$
  is at most $3$, $s_u$ and $s_v$ are connected by an edge in
  $G_u$. There therefore is a path of length at most $d=O(\log^3 n)$
  between $u$ and $v$ in $G_u$.

  To see that there are at most $|S|/(c\log^2 n)$ clusters, recall
  that any two nodes $u,v\in S'$ are at distance at least $c'\log^2n$
  from each other on $G$. Because a path of length $d$ in $G_S$
  implies a path of length at most $3d$ in $G$, $u$ and $v$ are
  therefore at distance at least $(c'/3)\log n$ in $G_S$. Each cluster
  is therefore of depth at least $(c'/6)\log^2n$ (in $G_S$) and it
  therefore contains at least $(c'/6)\log^2 n$ nodes of $S$. Because
  the clustering partitions the nodes of $S$, the number of clusters
  is therefore at most $6|S|/(c'\log n)$, which is at most
  $|S|/(c\log n)$ when choosing $c'$ sufficiently large.

  It remains to show that $G_S'$ is connected iff $G$ is
  connected. However, since two clusters are connected whenever there
  is an edge in $G_S$ connecting the two clusters, $G_S'$ is connected
  whenever $G_S$ is connected and this last claim therefore directly
  follows from \Cref{claim:basicCDS}.
\end{proof}

To complete the proof of \Cref{thm:mainThmCDS}, we next show that the deterministic spanner algorithm of \cite{DISC18_DomSet} can be efficiently applied on $G_S'$ to extend the set of cluster trees of the above clustering to a CDS of size at most $(3+\eps)|S|$. By its structure, the Baswana-Sen spanner algorithm \cite{baswana07} and also its derandomization in \cite{DISC18_DomSet} directly works on cluster graphs. As the edges of $G_S$ are paths of length at most $3$ in $G$, it is however not clear that we can efficiently communicate over the edges of $G_S$ (and thus of $G_S'$). 

\begin{proof}[{\bf Proof of \Cref{thm:mainThmCDS}:}]
  In \cite{DISC18_DomSet}, it
  is shown that on an $n$-node graph, a spanner with stretch $O(\log
  n)$ and $O(n\log^2 n)$ edges can be computed deterministically in
  $2^{O(\sqrt{\log n\log\log n})}$ in the \CONGEST model. The
  algorithm is based on a direct derandomization of the well-known
  randomized spanner construction by Baswana and Sen
  \cite{baswana07}. We next describe the steps of the
  algorithm that are necessary to understand its use in the context
  here. 

  Let $H=(V_H,E_H)$ be an $n$-node graph.  The algorithm consists of
  $k$ phases. At all times, the active nodes form disjoint
  clusters. At the beginning, each node is active and a cluster by
  itself. In each phase, each cluster is sampled with constant
  probability (say with probability $1/2$).\footnote{Note that the original
  algorithm by Baswana and Sen \cite{baswana07} allows to compute a
  spanner with stretch $2k-1$ for arbitrary $k$. The sampling
  probability in the algorithm is $n^{-1/k}$, which is equal to a
  constant for $k=\Theta(\log n)$.} At the end of the phase,
  only the sampled clusters survive and each node that is not in a
  sampled cluster either joins a cluster or becomes inactive. Let $v$
  be a node that is not in a sampled cluster. If $v$ has a neighbor
  $u$ in a sampled cluster, $v$ adds the edge $\set{u,v}$ to the
  spanner and joins the cluster of $u$ (if $v$ has more than one
  neighbor in sampled clusters, it chooses one arbitrarily). If $v$
  has no neighbor in a sampled cluster, it adds an edge to each
  neighboring cluster and becomes inactive. Because clusters are
  sampled with constant probability, a node that becomes inactive has
  at most $O(\log n)$ neighboring clusters, w.h.p., and it therefore
  also adds at most $O(\log n)$ edges to the spanner. In the last
  phase, there are at most $O(\log n)$ clusters and each active node
  adds an edge to each neighboring cluster. In each phase, each node,
  therefore adds at most $O(\log n)$ edges and thus, the total number
  of edges at the end is at most $O(n\log^2 n)$.\footnote{A more
    careful analysis actually shows that the number of edges is at
    most $O(n\log n)$.} It is clear that the resulting graph is
  connected. Assume now that the graph $H$ is actually a cluster
  graph, where the nodes of $H$ are represented by disjoint, connected
  clusters of diameter at most $d$ on some underlying graph
  $G$. Assume further that the clusters are connected by the edges
  of $G$ that connect nodes of different clusters. It is clear that at
  the cost of a factor $d$ in time, the above algorithm can also be
  run on $H$ by using the \CONGEST model on $G$.
 
  The only part that is randomized in the above spanner algorithm is
  the random sampling of the clusters. In \cite{DISC18_DomSet}, it is
  shown how this sampling can be replaced by a deterministic
  procedure. The requirements to run this deterministic sampling are
  also satisfied when running the algorithm on a cluster graph (in
  fact, since the spanner algorithm works on clusters anyways, the
  sampling in \cite{DISC18_DomSet} has to work for cluster graphs).

  The graph $G_S'$ is a cluster graph. However, the edges connecting
  the clusters are not edges of the underlying graph $G$, but they are
  edges of the graph $G_S$ (i.e., paths of length at most $3$ on
  $G$). It is not clear how to efficiently use all these $G_S$-edges
  in parallel in the \CONGEST model. In \cite{ghaffari14_CDS}, it is
  shown how to reduce the number of $G_S$ edges such that they can all
  efficiently be used in the \CONGEST model and such that the graph
  remains connected. We describe the method slightly applied to our
  context. 

  Recall that for each $v\in S'$, $C_v\subseteq S$ is the set of
  dominating set nodes in the cluster of node $v$. We need to connect
  any two clusters $C_u$ and $C_v$ by $G_S$-edges if there are two
  nodes $x\in C_u$ and $y\in C_v$ that are connected by $G_S$-edges
  (i.e., by a path of length at most $3$ in $G$). Note that we do not
  necessarily need to add one of $G_S$-edges that connectes $C_u$ with
  $C_v$ directly. We choose paths of length at most $3$ to connect the
  clusters as follows.

  \begin{enumerate}
  \item First of all, for any  two nodes $x\in C_u$ and $y\in C_v$
    that are neighbors in $G$, we add the edge $\set{x,y}$.
  \item For each node $w\in V\setminus S$, let $k(w)$ be the number of
    different clusters such that $w$ has a direct edge to some node
    $u\in S$ of the cluster. Each node with $k(w)\geq 1$ picks one
    neighbor for each adjacent cluster. Let these neighbors be
    $w_1,\dots,w_{k(w)}$. If $k(w)\geq 2$, node $w$ adds the $2$-hop paths
    $(w_{i},w,w_{i+1})$ as a $G_S$-edge between $w_i$ and $w_{i+1}$
    for each $i\in\set{1,\dots,k(w)-1}$.
  \item Let $w,w'\in V\setminus S$ be two $G$-neighbors with $k(w)\geq 1$
    and $k(w')\geq 1$. The two nodes $w,w'$ add the $3$-hop paths
    $(w_1,w,w',w'_{k(w')})$ and $(w_1',w',w,w_{k(w)})$
  \end{enumerate}
  
  It is clear that the added paths maintain the connectivity of the
  graph $G_S'$. Further, the paths are chosen such that each edge of
  $G$ is used in at most $2$ paths. We can therefore send $O(\log
  n)$-bit messages over all these paths in $O(\log n)$ rounds in the
  \CONGEST model on $G$. It is therefore now possible to efficiently
  apply the deterministic spanner algorithm of
  \cite{DISC18_DomSet}. The algorithm adds at most $O(|S'|\log^2|S'|)$
  $G_S$ edges and thus also at most $O(|S'|\log^2|S'|)$ nodes to the
  CDS. As the number of nodes in $S'$ is at most $|S|/(c\log^2 n)$,
  we have $O(|S'|\log^2|S'|)=O(\eps|S|)$ for a given constant $\eps>0$
  and a sufficiently large constant $c$. This completes the proof.
\end{proof}


\section{Conclusions}
\label{sec:conclusions}
We have efficient deterministic \CONGEST model algorithms to compute
almost optimal approximate solutions for the minimum dominating set
and the minimum connected dominating set problems. To keep things as
simple as possible, our algorithms are only given for the standard
(unweighted) dominating set problem. However two possible
generalizations are worth mentioning. The dominating set problem is a
special case of set cover problem, where one is given a collection of
finite sets that collectively cover some set of elements $X$ and the
objective is to choose a smallest possible number of those set, so
that still all elements in $X$ are covered. When considering the set
cover problem in a distributed context, it is usually modeled as a
graph by defining a node for each set and for each element and by
adding an edge between a set and an element node if and only if the
respective element is contained in the respective set (see, e.g.,
\cite{Kuhn06}). It is not hard to see that our algorithms can also be
(almost directly) applied to the more general set cover problem.

Another possible generalization would be to extend our algorithms to
also work for the weighted version of the dominating set
problem. Here, each node has a weight and the objective is to select a
dominating set of minimum total weight. We believe that our rounding
method would also work more or less in the same way for the weighted
dominating set problem (or for the weighted set cover problem). The
part of the algorithm that would require most care are the gradual
rounding steps, where we need to reduce the degrees of the constraint
nodes in order to obtain efficient algorithms. The partition of a
constraint into several lower order constraints seems less clear when
weights are involved. We should remark that while we believe that our
MDS algorithms can be generalized to also obtain good approximate
solutions for the weighted MDS problem, the same is not true for the
weighted connected dominating set problem. In \cite{ghaffari14_CDS},
Ghaffari shows that the weighted CDS problem has an
$\tilde{\Omega}(D + \sqrt{n})$ lower bound in the \CONGEST model on
graphs with diameter $D$. The paper also gives a
(randomized) \CONGEST algorithm that approximates the CDS problem in
time $\tilde{O}(D + \sqrt{n})$.

\bibliographystyle{alpha}
\bibliography{references}

\appendix


\end{document}